\documentclass[
journal]{IEEEtran}		
\usepackage{rotating}
\usepackage{amsmath}
\ifCLASSINFOpdf
\usepackage{graphicx,amssymb}
\usepackage{algorithm}
 \usepackage{algorithmic,slashbox}

\DeclareGraphicsExtensions{.pdf,.jpg,.jpeg,.png}
\else
\usepackage[dvips]{graphicx}
\usepackage{epsfig}
\DeclareGraphicsExtensions{.eps}
\fi

\usepackage{pdfsync}
\usepackage{amsthm}
\usepackage{color}

\newtheorem{theorem}{Theorem}[section]
\newtheorem{lemma}[theorem]{Lemma}
\newtheorem{proposition}[theorem]{Proposition}
\newtheorem{definition}[theorem]{Definition}
\newtheorem{corollary}[theorem]{Corollary}

\newtheorem{remark}[theorem]{Remark}

\newtheorem{example}[theorem]{Example}

\def\C{ {\mathbb C} }
\def\R{ {\mathbb R} }
\def\Z{ {\mathbb Z} }

\newcommand{\beq}{\begin {equation}}
\newcommand{\eeq}{\end{equation}}
\newcommand{\brem}{\begin {remark}}
\newcommand{\erem}{\end{remark}}

\numberwithin{equation}{section}
\ifCLASSINFOpdf
\else
\fi

\begin{document}

\title{Phase Retrieval of Real-Valued  Signals  in a   Shift-Invariant Space}
\author{Yang Chen,~ Cheng Cheng, ~ Qiyu Sun ~and~ Haichao Wang
\thanks{The project is partially supported by NSF (DMS-1412413).}
\thanks{Chen is with the Department of Mathematics,
Hunan Normal  University,
Changsha 410081, Hunan, China (yang\_chenww123@163.com).}
\thanks{Cheng and Sun are with the Department of Mathematics,
University of Central Florida,
Orlando 32816, Florida, USA (cheng.cheng@knights.ucf.edu; qiyu.sun@ucf.edu). }
\thanks{Wang was with  the Department of Mathematics, University of California at Davis (wanghaichao0501@gmail.com).
}}

\maketitle

\begin{abstract}
Phase retrieval arises in various fields of science and engineering
and it is well studied in a finite-dimensional setting.
In this paper, we consider an infinite-dimensional phase retrieval problem to reconstruct
real-valued signals
 living in a shift-invariant space from its phaseless samples taken either on the whole line or on a set with finite sampling rate.
We find the equivalence between nonseparability of signals in a linear space and its phase retrievability  with phaseless samples
taken  on the whole line.
For  a spline signal of order $N$, we show that it can be well approximated,
up to a sign, from its noisy phaseless samples taken on a  set
with sampling rate $2N-1$. We propose an algorithm to  reconstruct nonseparable signals in a shift-invariant space generated by a compactly supported continuous function $\phi$. The proposed algorithm is robust against bounded sampling noise and it could be implemented in a distributed manner.%

\end{abstract}

\IEEEpeerreviewmaketitle

\section{Introduction}

Phase retrieval plays important roles in signal/image/speech processing
 (\cite{F78}--\cite{jaganathany15}). It  reconstructs
a signal of interest from its magnitude measurements.
  The underlying   recovery problem is possible to be solved only if we
 have additional information about the signal.

The phase retrieval problem of finite-dimensional signals has received considerable attention in recent years
(\cite{BCE06}--\cite{W14}).
 In the finite-dimensional setting,
 a fundamental problem in phase retrieval is whether and how
  a vector  ${\bf x}\in  \R^d$ (or $\C^d$)
  can be reconstructed from its magnitude measurements  ${\bf y}=|{\bf A} {\bf x}|$, where
 ${\bf A}$ is a measurement  matrix. The phase  retrievability has been
fully characterized via the measurement matrix ${\bf A}$ (\cite{BCE06,  W14, Bandeira14}),
and many  algorithms have been proposed to reconstruct the  vector ${\bf x}$ from
its magnitude measurements  ${\bf y}$ (\cite{F78,  schechtman15,  CSV12, candes13, candes15,  gerchberg72, netrapalli15}).

The phase retrieval problem in an infinite-dimensional space
 is fundamentally different from a finite-dimensional setting.
There are several papers devoted to that topic
(\cite{T11}--\cite{cahill15}).
Thakur  proved in \cite{T11} that real-valued  bandlimited signals could be reconstructed from their phaseless samples  taken
 at more than twice  the Nyquist rate.  The above result was extended to complex-valued bandlimited signals by Pohl,  Yang and  Boche in \cite{volker14}
with samples taken at more than four times  the Nyquist rate.
Recently, the phase retrievability
of signals living in a principal
shift-invariant space
 was studied by Shenoy,  Mulleti and Seelamantula
in
\cite{shenoy16} when only magnitude  measurements of their frequency
are available.

Shift-invariant spaces
have
been widely used in  sampling theory, wavelet theory, approximation theory
 and
signal processing, see \cite{jia92}--\cite{AST05}
 and references therein.
In this paper, we consider the phase retrieval problem for  real-valued signals in a principal shift-invariant space
\begin{equation} \label{vphi.def} 
V(\phi):=\Big\{\sum_{k\in \Z}
 c(k) \phi(t-k): \ c(k)\in \R\Big\},\end{equation}
where the generator $\phi$ is a real-valued continuous
function with compact support.
Our model of the generator $\phi$ is the B-spline $B_N$ of order $N\ge 1$, which is obtained
by convoluting the indicator function $\chi_{[0,1)}$ on the unit interval  $N$ times,
\begin{equation}\label{bspline.def} B_N=\underbrace{\chi_{[0,1)}*\cdots*\chi_{[0,1)}}_N.\end{equation}

\subsection{Contribution}

In this paper, we show in Theorem \ref{separable.tm}  that a  real-valued signal $f\in V(\phi)$
 is determined, up to a sign,  from its magnitude $|f(t)|, t\in \R$, if and only if   $f$ is nonseparable, i.e.,  it is not the sum of two nonzero signals in $V(\phi)$
  with their supports being essentially disjoint.
  As an application of Theorem \ref{separable.tm}, we conclude that for any  shift-invariant space  $V(\phi)$ with continuous generator $\phi$ having compact support, not all signals in $V(\phi)$ could be determined, up to a sign,  from its magnitude $|f(t)|, t\in \R$, cf. \cite[Theorem 1]{T11} for the shift-invariant space  generated by the sinc function $\frac{\sin \pi t}{\pi t}$.

Phase retrieval in a shift-invariant space is a  nonlinear
 sampling  and reconstruction problem  (\cite{landau61, dvorkind08,   sun14}).
  In this paper, we show in Theorem \ref{realX.thm} and
  Corollary \ref{splineuniform.cor} that  a nonseparable spline signal in $V(B_N)$ is determined, up to a sign,  from its phaseless samples taken on the shift-invariant set
\begin{equation} \label{y1.def}
Y_1:=X+\Z,\end{equation}
where $N\ge 2$ and $X$ contains  $2N-1$ distinct
points in $(0,1)$.

The  set $Y_1$ in \eqref{y1.def} 
has sampling rate $2N-1$,
which is 
larger
than the sampling rate needed for the phase retrievability of band-limited signals \cite[Theorem 1]{T11}. Let
 \begin{equation}\label{supportlength.def}
 N=\min_{N_2, N_1\in \Z} \{N_2-N_1, \phi \ {\rm vanishes\ outside}  \ [N_1, N_2]\}\end{equation}
 be the support length of the generator $\phi$, which is the same as the order $N$ for the B-spline generator $B_N$.
A natural question is whether any nonseparable signal in the shift-invariant space $V(\phi)$ can be reconstructed
from its phaseless samples taken on a set with sampling rate less than $2N-1$.
From Example \ref{hat.example2} we see that any nonseparable linear spline signal in $V(B_2)$ can be determined, up to a sign,
  from its phaseless samples taken on the set $(\{x_1, x_2\}+\Z )\cup \{x_3\}$ with  sampling rate $2$, where $x_1, x_2, x_3$ are three distinct points in $(0,1)$.
 In Theorem \ref{realcase.thm}, 
we consider the phase retrieval problem of a nonseparable signal in the shift-invariant space $V(\phi)$
 from its phaseless samples taken on a nonuniform set
 \begin{equation}\label{yinfinty.def}
 Y_\infty:= X\cup (\Gamma+\Z_+)\cup (\Gamma^*+\Z_-)\end{equation}
 with sampling rate $N$,
 where $\Z_\pm$ is the set of all positive/negative integers, and the sets $\Gamma=\{\gamma_1, \ldots, \gamma_N\}$ and
 $ \Gamma^*=\{\gamma_1^*, \ldots, \gamma_N^*\}$ are contained in $X=\{x_1, \ldots, x_{2N-1}\}\subset (0,1)$.

Stability of phase retrieval is of central importance.
The reader may refer to \cite{Bandeira14, EM12,  BZ14, BW15} for phase retrieval in finite-dimensional setting and \cite{suntang15} for nonlinear frames. In this paper, we consider the scenario that phaseless samples taken on  a sampling set
 \begin{equation}\label{noise.samp}
Y_L=\Big(X\cup\Big\{\Gamma+l,\Gamma^\ast-l': 1\le l, l'\le \frac{L-1}{2}\Big\}\Big)+L\Z
\end{equation} are corrupted,
\begin{equation*}\label{noisedata.hat1.intro}
z_{\pmb\epsilon}(y )=|f(y)|^2+{\pmb \epsilon}(y), \  y\in Y_L,
 \end{equation*}
where  $L$ is an odd integer, $f$ is a nonseparable signal in $V(\phi)$,  and  additive noises
 ${\pmb \epsilon}= ({\pmb \epsilon}(y))_{y\in Y_L}$
have the noise level
$|{\pmb \epsilon}|=\sup\{|{\pmb \epsilon}(y)|:\ y\in Y_L\}$.
 In Theorem \ref{realother.thm}, we  establish the stability  of phase retrieval in the above scenario.

The set $Y_L$ in \eqref{noise.samp}
has sampling rate $N+(N-1)/L$. It becomes the shift-invariant set $X+\Z$ in \eqref{y1.def}
for $L=1$. Then as an application of
Theorem \ref{realother.thm},
any nonseparable spline signal in $V(B_N)$ can be reconstructed, up to a sign, approximately from its noisy phaseless samples on $X+\Z$.
 The nonuniform sampling set $Y_\infty$ in \eqref{yinfinty.def} can be interpreted as the limit of
the sets $Y_L$   as $L$ tends to infinity.  Due to the exponential  decay requirement  \eqref{realother.thm.eq2} about $L$ on the noise level,  we  cannot obtain from Theorem  \ref{realother.thm} that 
any nonseparable signal in  the shift-invariant space $V(\phi)$ could be well approximated, up to a sign,  when only its noisy phaseless samples on the nonuniform  set $Y_\infty$  are available.

 Many algorithms have been proposed to solve the phase retrieval problem in finite-dimensional setting  (\cite{F78, schechtman15, CSV12, candes13, candes15,  gerchberg72, netrapalli15}).  In this paper, we propose the MEPS algorithm
 to find an approximation  $f_{\pmb \epsilon}$ of a nonseparable signal  $f\in V(\phi)$
 when its noisy phaseless samples  $(z_{\pmb \epsilon}(y))_{y\in Y_L}$ are available.
 The MEPS algorithm 
  contains four steps: minimization, extension,  phase adjustment and sewing.
  Our numerical simulations indicate
 that the MEPS algorithm is  robust against bounded  additive noises ${\pmb \epsilon}$,  and the error  between
 the reconstructed signal $f_{\pmb \epsilon}$ and the original signal $f$  is $O(\sqrt{|\pmb \epsilon|})$.

\subsection{Organization}
 The paper is organized as follows.  In Section \ref{phase.section}, we
 characterize the phase retrievability of a real-valued signal $f$ in  a linear space
 from its  magnitude $|f(t)|, t\in \R$. We also provide several equivalent statements for
 the phase  retrievability of a signal in the shift-invariant space $V(\phi)$ when its phaseless samples on the shift-invariant set $X+\Z$
 in \eqref{y1.def}
are available  only.  In Section \ref{phase2.section},
we present an illustrative example of the phase retrieval problem for linear spline signals, and we
prove that  any nonseparable signal $f$ in the shift-invariant space $V(\phi)$ could be determined, up
to a sign, from its phaseless samples  $|f(t)|$ taken on the nonuniform sampling set $Y_\infty$ in \eqref{yinfinty.def}.
In  Section  \ref{stable.section}, we propose 
 the  MEPS algorithm
to reconstruct a nonseparable signal in $V(\phi)$ from its noisy phaseless samples on $Y_L$ in \eqref{noise.samp},
and we use it to establish the stability of the phase retrieval problem.
In Section \ref{simulation.section}, we present some simulations to demonstrate the stability of the proposed MEPS algorithm.
Even though the stability requirement
 \eqref{realother.thm.eq2} in Theorem \ref{realother.thm} is not met for large $L$, the MEPS algorithm still has
 high success rate to save phases of  nonseparable signals in $V(\phi)$.
 All proofs are included in appendices.

\section{Phase retrievability and  nonseparability}\label{phase.section}

In this section, we consider the problem when a signal $f$ in a shift-invariant space
can be recovered, up to a sign,  from its magnitude measurements $|f(t)|, t\in S$,
where $S$ is either the whole line $\R$ or a shift-invariant set $X+\Z$. 

\begin{definition}  Let $V$ be a linear space of real-valued continuous signals on the  real line $\R$.
A  signal $f\in V$ is said to be separable if
there  exist nonzero  signals  $f_1$ and $f_2$ in $V$ such that
\begin{equation}\label{separable.condition}
f=f_1+f_2\ \ {\rm and}\  \ f_1 f_2=0.\end{equation}
\end{definition}

The set of all nonseparable signals in $V$  contains the zero signal. It is  a cone of $V$ but not a convex set in general.
A  separable signal $f\in V$ is the sum of two nonzero signals $f_1$ and $f_2\in V$ with their supports being essentially disjoint.
Then it cannot be recovered, up to a sign, from its magnitude measurements $|f|$, since
$|f|=|f_1+f_2|=|f_1-f_2|$ and $f\ne \pm (f_1-f_2)$. 
In the following theorem,  we show that  the converse is  true.

\begin{theorem}\label{separable.tm}  Let $V$ be a linear space of real-valued continuous signals on the  real line $\R$.
Then a  signal $f\in V$ is determined, up to a sign,  by its magnitude measurements $|f(t)|, t\in \R$, if and only if
$f$ is nonseparable.
\end{theorem}

Observe that all bandlimited signals are nonseparable, as they are analytic on the real line. Therefore, by  Theorem \ref{separable.tm},
 we have the following result about bandlimited signals, cf.  \cite[Theorem 1]{T11}.

\begin{corollary}
 Any real-valued bandlimited  signal is determined, up to a sign,  by its magnitude measurements on the real line.
\end{corollary}

Let $\phi$ be a real-valued generator of the shift-invariant space $V(\phi)$, and  $N$ be its support length  given in \eqref{supportlength.def}.
Without loss of generality, 
we  assume that  
\begin{equation}\label{supportlength.def2}
\phi(t)=0\ {\rm  for\ all} \ t\not\in [0, N],\end{equation}
otherwise replacing $\phi$ by $\phi(\cdot-N_0)$ for some $N_0\in \Z$.
Clearly,  $\phi(t)-\phi(t-N)$ is a separable signal in $V(\phi)$. Then  from Theorem \ref{separable.tm} we obtain

 \begin{corollary} Let 
 $\phi$ be a 
  continuous function with compact support. Then
 not all signals $f$  in $V(\phi)$ can be determined, up to a sign,  by their magnitude measurements $|f(t)|, t\in \R$.
 \end{corollary}

Next, we   discuss the nonseparability of signals in a  shift-invariant space $V(\phi)$.
 For the case that $N=1$ (i.e., the generator $\phi$ is supported on $[0, 1]$),
  one may  verify that a 
 signal $f\in V(\phi)$
is nonseparable if and only if there exists an integer $k_0$ such that
\begin{equation}\label{N=1.function}
f(t)=c({k_0})\phi(t-k_0)\ \ {\rm  for\ some} \ \  c({k_0})\in \R.\end{equation}
  This implies that any nonseparable signal in $V(\phi)$  can be recovered, up to a  sign, from its phaseless samples taken on the shift-invariant set $t_0+\Z$, where $t_0\in (0,1)$ is so chosen that $\phi(t_0)\ne 0$.
So,  from now on, we consider the phase retrieval problem  only for  signals in $V(\phi)$ with
the support length $N$ of the generator $\phi$ satisfying
\begin{equation}\label{N.eq}
N\ge 2.
\end{equation}

\smallskip
Before characterizing the  nonseparability (and hence phase retrievability by Theorem \ref{separable.tm})
of signals in a shift-invariant space,  let us consider  nonseparability of piecewise linear signals.

\begin{example} \label{hat.example} {\rm
Due to the interpolation property of the B-spline  $B_2$  of
order $2$,  piecewise linear signals $f \in V(B_2)$ have the following expansion,
$$f(t)=\sum_{k\in \Z} f(k+1) B_2(t-k).$$
Therefore $f\in V(B_2)$ is separable if and only if there exist  integers $k_0<k_1<k_2$ such that
$f(k_0)f(k_2)\ne 0$ and $f(k_1)=0$. Thus the separable signal
$$f=\hskip-0.08in\sum_{k\le k_1-2} f(k+1)B_2(t-k)+ \hskip-0.03in\sum_{k\ge k_1} f(k+1) B_2(t-k)=:f_1+f_2,$$
is the sum of two nonzero signals $f_1, f_2\in V(B_2)$ supported in $(-\infty, k_1]$ and $[k_1, \infty)$ respectively.
} \end{example}

 In the following theorem, we extend the support separation property  in Example \ref{hat.example}
 to separable signals in a shift-invariant space.  


\begin{theorem}\label{realX.thm} Let $\phi$ be a  real-valued
 continuous  function satisfying
\eqref{supportlength.def2} and \eqref{N.eq},
$X:=\{x_m, 1\le m\le 2N-1\}\subset (0, 1)$, and let
$f(t)=\sum_{k\in \Z} c(k) \phi(t-k)$ be a nonzero real-valued
signal in $V(\phi)$.
  If
all $N\times N$ submatrices of
\begin{equation}\label{Phi.def}
\Phi=\big(\phi(x_m+n)\big)_{1\le m\le 2N-1, 0\le n\le N-1}\end{equation}
are nonsingular, then the following statements are equivalent.
\begin{itemize}

\item [{(i)}] The signal $f$ is nonseparable.

\item [{(ii)}] $\sum_{l=0}^{N-2} |c({k+l})|^2\ne 0$ for all $K_-(f)-N+1<k<K_+(f)+1$, where
$K_-(f)=\inf\{k, c(k)\ne 0\}$ and $K_+(f)=\sup\{k, c(k)\ne 0\}$.

\item [{(iii)}] The signal $f$ is determined, up to a sign,  from  its phaseless samples $|f(t)|, t\in X+\Z$, taken on the shift-invariant set $X+\Z$.
\end{itemize}
\end{theorem}

The nonsingularity of $N\times N$ submatrices of the matrix $\Phi$ in \eqref{Phi.def}, i.e.,
$\|(\Phi_N)^{-1}\|<\infty$,
 is
 also known as its full sparkness (\cite{donoho03, ACM}), where
  \begin{eqnarray} \label{realother.thm.eq1}
\hskip-0.02in\|(\Phi_N)^{-1}\| & \hskip-0.08in =  & \hskip-0.08in \sup_{1\le m_0<\ldots<m_{N-1}\le 2N-1}\nonumber\\
 & &  \Big\| \Big(\big(\phi(x_{m_l}+n)\big)_{0\le l , n\le N-1}\Big)^{-1}\Big\|,
\end{eqnarray}
and  $\|A\|=\sup_{\|x\|_2=1} \|Ax\|_2$ for a matrix $A$.

Consider the matrix $\Phi$ with its generating function  $\phi$ being the continuous solution  of a refinement equation,
\begin{equation}\label{refinement.def}
\phi(t)=\sum_{n=0}^N a(n) \phi(2t-n)\ \  {\rm and}  \ \  \int_{-\infty} ^\infty \phi(t)dt=1 ,\end{equation}
 where $\sum_{n=0}^N a(n)=2$.
Under the assumption that
\begin{equation}\label{rip.con}
\sum_{n=0}^N a(n) z^n= (1+z) Q(z)\end{equation}
for some polynomial $Q$ having positive coefficients,  it is known
that   the matrix $\Phi$ in \eqref{Phi.def} is of full spark whenever $x_m\in (0, 1), 1\le m\le 2N-1$, are distinct  (\cite{goodman92, goodman04}).
It is well known that the B-spline $B_N$ of order $N$
satisfies the refinement equation \eqref{refinement.def} with $Q(z)$ in \eqref{rip.con} given by $2^{-N+1}(1+z)^{N-1}$.
This together with  Theorem \ref{realX.thm} implies the following result for  spline signals.

\begin{corollary}\label{splineuniform.cor}
Let  $X$ contain $2N-1$ distinct points in $(0, 1)$. Then any nonseparable spline signal in $V(B_N)$
is determined, up to a sign,  from its phaseless samples taken on the shift-invariant set $X+\Z$.
\end{corollary}

The full sparkness of the matrix $\Phi$ in \eqref{Phi.def} implies that $\phi$ has linearly independent shifts, i.e.,
 the linear  map from sequences
$(c(k))_{k=-\infty}^\infty$ to signals $\sum_{k=-\infty}^\infty c(k) \phi(t-k)\in V(\phi)$ is one-to-one (\cite{jia92, ron89, S10}).  Conversely, if  $\phi$ is the continuous solution of a refinement equation \eqref{refinement.def} with linearly independent shifts, then  $\Phi$ in \eqref{Phi.def} is of full spark for almost all $(x_1, \ldots, x_{2N-1})\in (0,1)^{2N-1}$, see  \cite[Theorem A.2]{S10}.

For a signal $f=\sum_{k\in \Z}
 c(k) \phi(t-k)\in V(\phi)$, define
\begin{equation}\label{sf.def}
S_f=\inf_{K_-(f)-N+1<k<K_+(f)+1} \sum_{l=0}^{N-2} |c(k+l)|^2.\end{equation}
By  the second statement in Theorem \ref{realX.thm}, we obtain
that $S_f=0$ for any separable signal $f\in V(\phi)$, and  that $S_f>0$  for any nonseparable signal $f\in V(\phi)$ with finite duration.
So we may use $S_f$ to estimate
 how far  away a  nonseparable signal $f$  from the set of all separable signals in $V(\phi)$, cf. Theorem \ref{realother.thm}.

\section{Phaseless oversampling}\label{phase2.section}

 A discrete set $I\subset \R$ is said to have sampling rate $D(I)$ if
\begin{equation}\label{samplingrate.def}
D(I)=\lim_{b-a\to \infty}   \frac{\#(I\cap [a, b])}{b-a}<\infty,
\end{equation}
where   $\#(E)$ is the cardinality of a set $E$.
Let $\phi$ be the continuous function satisfying \eqref{supportlength.def2}, \eqref{N.eq} and \eqref{Phi.def}.  It follows immediately from
Theorem \ref{realX.thm}  that  nonseparable signals in  $V(\phi)$
can be fully recovered, up to a sign,  from their phaseless samples taken on the shift-invariant set $X+\Z$ with sampling rate $2N-1$,
 which is larger than the sampling rate   required for recovering bandlimited signals \cite[Theorem 1]{T11}.
A natural question  is to find necessary/sufficient conditions on a  set $I$ such that any nonseparable signal in $V(\phi)$ can be reconstructed from its phaseless samples taken on $I$.

In this section, we first introduce a necessary condition on the sets $I$.

\begin{theorem}\label{necessary.thm}
Let $\phi$ be a  real-valued
 continuous  function satisfying
\eqref{supportlength.def2}, \eqref{N.eq} and \eqref{Phi.def}, and let $I$ be a discrete set with sampling rate $D(I)$. If all nonseparable signals in $V(\phi)$ can be determined, up to a sign, from their phaseless samples taken on the set $I$, then the sampling rate $D(I)$ is at least one,
\begin{equation} \label{necessary.thm.eq1} D(I)\ge 1.
\end{equation}
\end{theorem}

The lower bound estimate \eqref{necessary.thm.eq1}
is  smaller than the sampling rate   required for recovering bandlimited signals \cite[Theorem 1]{T11}.
So one may think that it can be improved. However as indicated in the example below,
the  lower bound estimate \eqref{necessary.thm.eq1} is optimal  if the requirement \eqref{Phi.def} on the  generator $\phi$ is dropped.

 \begin{example}{\rm
Let $\varphi_0$ be a continuous function supported in $[0, 1/2]$
and set $\varphi_N(t)=\varphi_0(t)-\varphi_0(t-N+1/2), N\ge 1$.
Similar to \eqref{N=1.function}, one may verify that
 a signal $f$ in $V(\varphi_N)$ is nonseparable
if and only if there exists $k_0\in \Z$ such that
$$f(t)=c({k_0})\varphi_N(t-k_0)\ \ {\rm  for\ some} \  \ c({k_0})\in \R.$$
Hence given any $t_0\in (0, 1/2)$ with $\varphi_0(t_0)\ne 0$,  all nonseparable signals in $V(\varphi_N)$ can be reconstructed, up to a sign,  from their phaseless samples taken on the set
$t_0+\Z$ with sampling rate one.
}\end{example}

In this section,
we next show that  nonseparable signals in  $V(\phi)$
are determined, up to a sign,  from their phaseless samples taken on a  set with sampling rate $N$.  
Before stating the result, let us briefly discuss an  example of phaseless oversampling. 
\begin{example}
\label{hat.example2} {\rm  (Continuation of Example \ref{hat.example}) \  Let $k_0\in \Z$ and $f\in V(B_2)$ be a nonseparable  piecewise linear signal. 
One may verify that $3$ distinct points $k_0+x_1, k_0+x_2, k_0+x_3\in k_0+(0, 1)$ are enough to  determine $f(k_0)$ and $f(k_0+1)$ (hence $f(t), t\in k_0+[0,1]$), up to a phase, from  phaseless samples $|f(k_0+x_1)|, |f(k_0+x_2)|$ and $|f(k_0+x_3)|$. Particularly, solving
$$|f(k_0)(1-x_i)+x_i f(k_0+1)|^2=|f(k_0+x_i)|^2,\ i=1, 2, 3$$ gives
\begin{equation*}\label{f0.solution}
|f(k_0)|^2= \frac{\left |\begin{array}{ccc}  |f(k_0+x_1)|^2 & x_1(1-x_1) & x_1^2\\
|f(k_0+x_2)|^2 & x_2 (1-x_2) & x_2^2\\
|f(k_0+x_3)|^2 & x_3 (1-x_3) & x_3^2\end{array}\right| }{(x_2-x_1)(x_3-x_1)(x_3-x_2)},
\end{equation*}
\begin{equation*} \label{f0f1.solution}
2f(k_0)f(k_0+1)= \frac{\left |\begin{array}{ccc} (1-x_1)^2 &  |f(k_0+x_1)|^2  & x_1^2\\
(1-x_2)^2 & |f(k_0+x_2)|^2  & x_2^2\\
(1-x_3)^2 & |f(k_0+x_3)|^2  & x_3^2\end{array}\right| }{(x_2-x_1)(x_3-x_1)(x_3-x_2)},
\end{equation*}
and
\begin{equation*}  \label{f1.solution} |f(k_0+1)|^2= \frac{\left |\begin{array}{ccc} (1-x_1)^2 & x_1(1-x_1) &  |f(k_0+x_1)|^2 \\
(1-x_2)^2 & x_2(1-x_2) & |f(k_0+x_2)|^2 \\
(1-x_3)^2 & x_3 (1-x_3) & |f(k_0+x_3)|^2 \end{array}\right| }{(x_2-x_1)(x_3-x_1)(x_3-x_2)}.
\end{equation*}
For the case that 
at lease  one of
two evaluations $f(k_0)$ and $f(k_0+1)$ is nonzero,
\begin{equation}\label{piecewise.eq}
f(k_0+2) = \left\{ \begin{array}{ll} 0 & {\rm if}\ f(k_0+1)=0\\
 f(k_0+1)+ \triangle_{k_0}^+ & {\rm if} \  f(k_0+1)\ne 0,
 \end{array}\right.
\end{equation}
where the first equality follows from  nonseparability of the signal $f$, the second
one is obtained by solving
the equations \begin{equation}\label{k0+1.eq}
|f(k_0+1)(1-x_i)+ x_i f(k_0+2)|^2=|f(k_0+1+x_i)|^2, i=1, 2,\end{equation}
and
\begin{eqnarray*} \triangle_{k_0}^+ &\hskip-0.08in = & \hskip-0.08in\frac{x_1^2 (|f(k_0+1+x_2)|^2- |f(k_0+1)|^2)}{2x_1x_2(x_1-x_2) f(k_0+1)}\nonumber\\
&\hskip-0.08in &\hskip-0.08in -\frac{x_2^2|f(k_0+1+x_1)|^2-|f(k_0+1)|^2)} {2x_1x_2(x_1-x_2) f(k_0+1)}.
\end{eqnarray*}
From \eqref{piecewise.eq} we see that two distinct points
$k_0+1+x_1, k_0+1+x_2\in k_0+1+(0,1)$ could sufficiently  determine $f(t), k_0+1\le t\le k_0+2$.

 For the case that 
   $f(k_0+1)= f(k_0)=0$,
solving \eqref{k0+1.eq}
yields
$$|f(k_0+2)|^2=\frac{|f(x_1+k_0+1)|^2+|f(x_2+k_0+1)|^2}{x_1^2+x_2^2}.$$
Then either  $f(t)=0$ for all $t\in [k_0, k_0+2]$  or the
phase of the signal $f$ on $[k_0, k_0+2]$ is determined up to the sign of nonzero evaluation $f(k_0+2)$.  Therefore, we can continue the above procedure to determine the signal $f$ on $[k_0,\infty)$ if there are two distinct points in intervals $k+(0, 1)$ for every $k\ge k_0+1\in \Z\backslash \{k_0\}$.

 Using the similar argument,  we can prove by induction on $k <k_0$ that the signal $f(t), t\in [k, \infty)$, can be
 determined, up to a sign, by its phaseless samples taken on $l-1+x_1$ and $l-1+x_2$, $k\le l<k_0$.
 By now, we conclude that a
nonseparable   signal in $V(B_2)$
could be  determined, up to a sign,  by its phaseless samples on  $(\{x_1, x_2\}+\Z)\cup\{x_3+k_0\}$, where $x_1, x_2, x_3\in (0, 1)$ are distinct
and $k_0\in \Z$. We remark that the additional point $x_3+k_0$ in the above phase retrievability is necessary in general. For  instance,
signals $f(t)\equiv 1/3$ and $g(t)=\sum_{k\in \Z} (-1)^{k}B_2(t-k)$ in $V(B_2)$ have the same magnitude  measurements  on  $\{1/3, 2/3\}+\Z$, but $f\ne \pm g$.
}\end{example}

 Finally, 
 we state the result on the phase retrieval of  
 %
 nonseparable signals in a shift-invariant space $V(\phi)$ with sampling rate $N$. 

\begin{theorem}\label{realcase.thm}
Let $\phi$ be a  real-valued  continuous  function satisfying
\eqref{supportlength.def2} and \eqref{N.eq}.
Take $X:=\{x_m, 1\le m\le 2N-1\}\subset (0, 1), \Gamma=\{\gamma_1, \ldots, \gamma_N\}\subset X$
and $\Gamma^*=\{\gamma_1^*, \ldots, \gamma_N^*\}\subset X$
so that  the matrix $\Phi$ in \eqref{Phi.def}
is of full spark,
 \begin{equation}\label{phinonzero.cond}
 \phi(\gamma_s)\ne 0 \   \ \text{for all} \ 1\le s \le N,
 \end{equation}
 and
 \begin{equation}\label{phinonzero*.cond}
 \phi(\gamma^*_s+N-1)\ne 0  \  \ \text{for all} \ 1\le s \le N.
 \end{equation}
Then for any $k_0\in \Z$, a nonseparable signal  in $V(\phi)$ is determined, up to  a sign,
from  its phaseless samples 
taken on $(X+k_0)\cup (\Gamma+k_0+\Z_+) \cup (\Gamma^{\ast}+k_0+\Z_-)$.
\end{theorem}

By the nonsingularity of any $N\times N$ submatrices of the matrix $\Phi$ in  \eqref{Phi.def}, there are at least $N$ distinct elements
$\gamma_1, \ldots, \gamma_N$ contained in $X$ such that
\eqref{phinonzero.cond} holds. Similarly there are
at least $N$ distinct elements
$\gamma_1^*, \ldots, \gamma_N^*\in X$ satisfying  \eqref{phinonzero*.cond}.
Therefore  \eqref{phinonzero.cond} and \eqref{phinonzero*.cond} hold for some $\Gamma, \Gamma^*\subset X$.

  The requirements
\eqref{phinonzero.cond} and \eqref{phinonzero*.cond} in Theorem \ref{realcase.thm}
 are met for any subsets $\Gamma, \Gamma^*\subset X$, provided that $\phi$ is a refinable function  with its symbol satisfying \eqref{rip.con} \cite{goodman92}.
 Therefore as an application of Theorem \ref{realcase.thm}, we have the following result for spline signals.

 \begin{corollary}\label{bsplinerealcase.cor}
Let $X$ contain $2N-1$ distinct points in $(0,1)$, and
 $\Gamma$ be a subset of $X$ of size $N$. Then  any nonseparable spline signal in $V(B_N)$
  is determined, up to a sign, from its phaseless samples taken
 on $(\Gamma+\Z) \cup X$.
\end{corollary}

\section{Stability of phase retrieval}
\label{stable.section}

Stability of phase retrieval is  of central importance, as
phaseless samples in lots of engineering applications are often corrupted.
In this section, we establish the stability of phase retrieval
 in a shift-invariant space, when
its phaseless samples taken on
the  set $Y_L$
  are  corrupted by  additive noises
 ${\pmb \epsilon}= (\pmb\epsilon(y))_{y\in Y_L}$,
 \begin{equation}\label{noisedata.hat1}
z_{\pmb\epsilon}(y )=|f(y)|^2+{\pmb\epsilon}(y), \  y\in Y_L,
 \end{equation}
where  $Y_L$ is given in
\eqref{noise.samp} for an odd integer $L$, and
${\pmb \epsilon}$ has the noise level $$|{\pmb \epsilon}|=\sup\{|\pmb\epsilon(y)|:\ y\in Y_L\}.$$

For $L=1$, it follows from Theorem \ref{realX.thm}  that nonseparable signals in $V(\phi)$ can be recovered, up to a sign,  from their  exact phaseless samples on $Y_L$.  By Theorem \ref{realcase.thm},  nonseparable signals with finite duration are  determined, up to a sign,  from their
exact phaseless samples on $Y_L$ with sufficiently large $L$.

\smallskip

To present an algorithm for phase retrieval in a noisy environment, we introduce four auxiliary
functions.
 Let $X$, $\Gamma$ and $\Gamma^\ast$ be as in Theorem \ref{realcase.thm}. Define
\begin{equation}\label{gGamma.def}
h_1({\bf e})=\frac{\left |\hskip-0.05in\begin{array}{cc}
 \sum_{n=1}^N |\phi(\gamma_n)|^2   & \hskip-0.05in\sum_{n=1}^N   \phi(\gamma_n) e(n)\\
\sum_{n=1}^N \phi(\gamma_n) e(n) &  \hskip-0.05in\sum_{n=1}^N  |e(n)|^2
\end{array}\hskip-0.05in\right|}
{\sum_{n=1}^N |\phi(\gamma_n)|^2}
\end{equation}
and
\begin{equation}\label{gGammastar.def}
 h_1^\ast ({\bf e})=\frac{\left |\hskip-0.05in\begin{array}{cc}
 \sum_{n=1}^N |\phi(\gamma_n^{\ast\ast})|^2   & \hskip-0.05in\sum_{n=1}^N   \phi(\gamma_n^{\ast\ast}) e(n)\\
\sum_{n=1}^N \phi(\gamma_n^{\ast\ast}) e(n) &  \hskip-0.05in\sum_{n=1}^N  |e(n)|^2
\end{array}\hskip-0.05in\right|}
{\sum_{n=1}^N |\phi(\gamma_n^{\ast\ast})|^2}, \end{equation}
where ${\bf e}=(e(1),\ldots,e(N))\in \R^N$ and $\gamma_n^{\ast\ast}=\gamma^*_n+N-1, 1\le n\le N$.  Define
\begin{equation}\label{gGamma2.def}
h_2({\bf e}_1, {\bf e}_2)=
\frac{\left |\hskip-0.05in\begin{array}{cc}
 \sum_{n=1}^N |\phi(\gamma_n)|^2   & \hskip-0.05in\sum_{n=1}^N  e_2(n)\\
\sum_{n=1}^N \phi(\gamma_n) e_1(n) &  \hskip-0.05in\sum_{n=1}^N  \frac{e_1(n) e_2(n)}{\phi(\gamma_n)}
\end{array}\hskip-0.05in\right|}{\sum_{n=1}^N |\phi(\gamma_n)|^2}\end{equation}
and
\begin{equation}\label{gGamma2star.def}
h_2^*({\bf e}_1, {\bf e}_2)=\frac{\left |\hskip-0.05in\begin{array}{cc}
 \sum_{n=1}^N |\phi(\gamma_n^{\ast\ast})|^2   & \hskip-0.05in\sum_{n=1}^N  e_2(n)\\
\sum_{n=1}^N \phi(\gamma_n^{\ast\ast}) e_1(n) &  \hskip-0.05in\sum_{n=1}^N  \frac{e_1(n) e_2(n)}{\phi(\gamma_n^{\ast\ast})}
\end{array}\hskip-0.05in\right|}{\sum_{n=1}^N |\phi(\gamma_n^{\ast\ast})|^2},\end{equation}
where ${\bf e}_i=(e_i(1),\ldots,e_i(N))\in \R^N$,  $i=1, 2$.

\medskip

Take a threshold  $M_0\ge 0$, we propose
 the following  algorithm
 to construct an approximation
 \begin{equation}\label{fepsilon.def}
 f_{\pmb \epsilon}(t)=\sum_{k\in \Z} c_{\pmb \epsilon}(k) \phi(t-k)\in V(\phi)\end{equation}
of the original signal
\begin{equation}\label{signal.def}
 f=\sum_{k\in\Z}c(k)\phi(x-k)\in V(\phi),
 \end{equation}
 when only  its noisy phaseless samples
  in \eqref{noisedata.hat1}
 are available.

\begin{itemize}

\item[{(i)}] For any $k'\in \Z$, 
we obtain an approximation
\begin{equation}\label{c.epsilon.0.def}
{\bf c}_{\pmb \epsilon, k'}^0=(c_{\pmb \epsilon, k'}^0(k))_{k\in \Z},
\end{equation}
 of the original amplitude vector $\pm {\bf c}$ on $[k'L-N+1, k'L]$ as follows.
Initialize $c_{\pmb \epsilon, k'}^0(k)=0$  for all $k\le k'L-N$ and $k\ge k'L+1$, and
let $c_{\pmb \epsilon, k'}^0(k), k'L-N+1\le k\le k'L$, be  solutions of the  minimization problem
\begin{equation} \label{phaselessminimization.1}
\hskip-0.1in\min \sum_{m=1}^{2N-1}  \Bigg|\Big|\sum_{k=k'L-N+1}^{k'L} c(k) \phi(x_{m,k'}-k)\Big|-\sqrt{z_{\pmb\epsilon}(x_{m,k'} )}\Bigg|^2,
\end{equation}
where $x_m\in X$ and  $x_{m,k'}=x_m+k'L, 1\le m\le 2N-1$, cf. (\cite{F78, gerchberg72, netrapalli15, qiu}). The support of ${\bf c}_{\pmb \epsilon, k'}^0$ is contained in $[k'L-N+1, k'L]$ for any $k'\in \Z$.
\item[{(ii)}] For any $k'\in \Z$, 
  define ${\bf c}_{\pmb \epsilon, k', 0}^{1/2}= {\bf c}_{\pmb \epsilon, k'}^{0}$
 and 
\begin{equation}
\label{c1/2.defold} {\bf c}_{\pmb \epsilon, k', l}^{1/2}=(c_{\pmb \epsilon, k', l}^{1/2}(k))_{k\in \Z}, \ 1\le l \le (L-1)/2,\end{equation}
 recursively by the following:
 \begin{itemize}
 \item  If  $|h_1({\pmb \alpha}_{\pmb \epsilon, l})|\le  M_0$,
we obtain ${\bf c}^{1/2}_{\pmb\epsilon,k', l}$ from ${\bf c}^{1/2}_{\pmb\epsilon,k', l-1}$ with
 replacing its $(k'L+l)$-th  component by
 \begin{equation}\label{c.half.leastsquare}
c^{1/2}_{\pmb\epsilon,k', l}(k'L+l)=\frac{\sum_{n=1}^{N}|\phi(\gamma_n)|\sqrt{z_{\pmb\epsilon}(\gamma_n+k'L+l)}}{\sum_{n=1}^{N}|\phi(\gamma_n)|^2},
\end{equation}
 where ${\pmb \alpha}_{\pmb \epsilon, l}=(\alpha_{\pmb \epsilon, l}(1), \ldots$, $\alpha_{\pmb \epsilon, l}(N))$ and
 \begin{equation*}\label{alpha.def}
\alpha_{\pmb\epsilon, l}(n)=\sum_{n'=1}^{N-1}c^{1/2}_{\pmb\epsilon,k', l-1}(k'L+l-n')\phi(\gamma_n+n')
\end{equation*}
for $1\le n\le N$.

 \item  If $|h_1({\pmb \alpha}_{\pmb \epsilon, l})|> M_0$,
 set
\begin{equation} \label{phaselessminimization.1n}
d_{\pmb\epsilon,k'}(k'L+l)=\frac{ h_2({\pmb \alpha}_{\pmb \epsilon, l}, {\pmb \eta}_{\pmb \epsilon, l})}
{ 2 h_1({\pmb \alpha}_{\pmb \epsilon, l})},
\end{equation}
 where ${\pmb \eta}_{\pmb \epsilon, l}=(\eta_{\pmb \epsilon, l}(1), \ldots, \eta_{\pmb \epsilon, l}(N))$
 is defined by
 $$\eta_{\pmb\epsilon, l}(n)=z_{\pmb\epsilon}(\gamma_n+k'L+l)-|\alpha_{\pmb\epsilon, l}(n)|^2, 1\le n\le N.$$
  For $1\le n\le N$, let
  \begin{eqnarray}\label{mewhalf.eq2}
  \delta_l(n) & \hskip-0.08in = & \hskip-0.08in  {\rm sgn} \Big(\sum_{m=1}^{N-1} c^{1/2}_{\pmb\epsilon,k', l-1}(k'L+l-m)\phi(\gamma_n+m)\nonumber\\
  &  \hskip-0.08in & \hskip-0.08in \quad + d_{\pmb\epsilon,k'}(k'L+l) \phi(\gamma_n)\Big)
 \end{eqnarray}
 and
 \begin{equation}\label{newhalf.eq3}
 \tilde {z}_{\pmb \epsilon}(k'L+l+\gamma_n)= \delta_l(n) \sqrt{z_{\pmb\epsilon}(\gamma_n+k'L+l)},
 \end{equation}
 where ${\rm sgn}(x)\in \{-1, 0, 1\}$ is the symbol of a real number $x$.
Now we obtain ${\bf c}^{1/2}_{\pmb\epsilon,k', l}$ from ${\bf c}^{1/2}_{\pmb\epsilon,k', l-1}$ by updating
its $(k'L+l-m)$-th terms by  
the unique solution $d(k'L+l-m), 0\le m\le N-1$, of the linear system,
  \begin{equation}  \label{newhalf.eq3}  
 \sum_{m=0}^{N-1} d(k'L+l-m)\phi(\gamma_n+m)=\tilde {z}_{\pmb \epsilon}(k'L+l+\gamma_n), 
\end{equation}
where $1\le n\le N$. One may verify that the support of ${\bf c}_{\pmb \epsilon, k',l}^{1/2}$ is contained in $[k'L-N+1, k'L+l]$ for any $k'\in \Z$ and $0\le l\le (L-1)/2$.

Finally define
\begin{equation}
\label{c1/2.def} {\bf c}_{\pmb \epsilon, k'}^{1/2}={\bf c}_{\pmb \epsilon, k', (L-1)/2}^{1/2}, \ k'\in \Z.\end{equation}


\end{itemize}

\item[{(iii)}]  Set $\gamma_n^{\ast\ast}=\gamma^*_n+N-1, 1\le n\le N$.
Define
${\bf c}_{\pmb \epsilon, k', 0}^1= {\bf c}_{\pmb \epsilon, k'}^{1/2}$ and
${\bf c}_{\pmb \epsilon, k', l'}^1, 1\le l'\le (L-1)/2$, recursively:
\begin{itemize}
\item  If   $|h_1^\ast({\pmb \alpha}^\ast_{\pmb \epsilon, l'})|\le  M_0$,
then we update ${\bf c}_{\pmb \epsilon, k', l'}^1$ from
${\bf c}_{\pmb \epsilon, k', l'-1}^1$ by replacing its $(k'L+1-N-l')$-th term with
  \begin{eqnarray} \label{III.11}
\hskip-0.35in  & \hskip-0.08in &  \hskip-0.38in c^1_{\pmb\epsilon,k', l'}(k'L+1-N-l')\nonumber \\
\hskip-0.35in & \hskip-0.08in  & \hskip-0.38in \ \  = \frac{\sum_{n=1}^{N}|\phi(\gamma^{\ast\ast}_n)|\sqrt{z_{\pmb\epsilon}(\gamma^{\ast\ast}_n+k'L-l')}}{\sum_{n=1}^{N}|\phi(\gamma^{\ast\ast}_n)|^2}, \end{eqnarray}
  where
 ${\pmb \alpha}^\ast_{\pmb \epsilon, l'}= (\alpha^\ast_{\pmb \epsilon, l'}(1),  \ldots$,   $\alpha^\ast_{\pmb \epsilon, l'}(N))$
is given by
\begin{equation*}
\alpha^\ast_{\pmb\epsilon, l'}(n)=\sum_{n'=0}^{N-2}c^1_{\pmb\epsilon,k', l'-1}({k'L-l'-n'})\phi(\gamma_n^{\ast}+n').
\end{equation*}

\item If $|h_1^\ast({\pmb \alpha}^\ast_{\pmb \epsilon, l'})|> M_0$, set 
  \begin{equation}\label{phaselessminimization.1n++}
\tilde d_{\pmb\epsilon,k'}(k'L+1-N-l')=\frac{ h_2^\ast({\pmb \alpha}^\ast_{\pmb \epsilon, l'}, {\pmb \eta}^\ast_{\pmb \epsilon, l'} )}{2h_1^\ast({\pmb \alpha}^\ast_{\pmb \epsilon, l'})},
\end{equation}
where
${\pmb \eta}^\ast_{\pmb \epsilon, l'}= (\eta^\ast_{\pmb \epsilon, l'}(1),  \ldots,   \eta^\ast_{\pmb \epsilon, l'}(N))$
is defined by
$$\eta^\ast_{\pmb\epsilon, l'}(n)=z_{\pmb\epsilon}(\gamma_n^{\ast\ast}+k'L-l')-|\alpha^\ast_{\pmb\epsilon, l'}(n)|^2, 1\le n\le N.$$
 For $1\le n\le N$, define
  \begin{eqnarray}\label{newone.eq2}
  \hskip-0.28in\tilde \delta_{l'}(n) & \hskip-0.08in = & \hskip-0.08in  {\rm sgn} \Big(\sum_{m=0}^{N-2} c^{1/2}_{\pmb\epsilon,k', l-1}(k'L-l'-m)\phi(\gamma_n^*+m)\nonumber\\
 \hskip-0.28in &  \hskip-0.08in & \hskip-0.08in \quad + \tilde d_{\pmb\epsilon,k'}(k'L+1-N-l') \phi(\gamma_n^{\ast\ast})\Big)
 \end{eqnarray}
 and
 \begin{equation}\label{newone.eq3-}
\hskip-0.18in \tilde {z}_{\pmb \epsilon}(k'L-l'+\gamma_n^*)= \tilde \delta_{l'}(n) \sqrt{z_{\pmb\epsilon}(k'L-l'+\gamma_n^*)}.
 \end{equation}
Now we get ${\bf c}^{1}_{\pmb\epsilon,k', l'}$ from ${\bf c}^{1}_{\pmb\epsilon,k', l'-1}$ by updating
its $(k'L-l'-m)$-th terms with  
the solution $d(k'L-l'-m), 0\le m\le N-1$, of the linear system:
  \begin{equation}  \label{newone.eq3}  
  \sum_{m=0}^{N-1} d(k'L-l'-m)\phi(\gamma_n^*+m)=\tilde {z}_{\pmb \epsilon}(k'L-l'+\gamma_n^*), 
\end{equation}
where $1\le n\le N$. From the above construction,  ${\bf c}_{\pmb \epsilon, k',l'}^{1}$ is supported in $[k'L-N+1-l', k'L+(L-1)/2]$ for any $k'\in \Z$ and $0\le l'\le (L-1)/2$.

Finally define the following approximation
\begin{equation}
\label{c1.def} {\bf c}_{\pmb \epsilon, k'}^{1}={\bf c}_{\pmb \epsilon, k', (L-1)/2}^{1}, \ k'\in \Z, \end{equation}
of the original amplitude vector $\pm {\bf c}$ on $[k'L-N-(L+1)/2, k'L+(L-1)/2]$.

\end{itemize}

\item[{(iv)}] Adjust  phases of ${\bf c}_{\pmb\epsilon, k'}^1, k'\in \Z$, by
\begin{equation} \label{phaselessminimization.2}
{\bf c}_{\pmb\epsilon, k'}^2= \tilde\delta_{k'} {\bf c}_{\pmb\epsilon, k'}^1,
\end{equation}
where $\tilde \delta_{k'}\in \{-1, 1\}$ are so chosen that 
\begin{equation} \label{phaselessminimization.3}
 \langle {\bf c}_{\pmb\epsilon, k'}^2, {\bf c}_{\pmb\epsilon, k'+1}^2\rangle\ge 0 \ {\rm for  \ all} \ k'\in \Z.
\end{equation}

\item[{(v)}] Define ${\bf c}_{\pmb\epsilon}=(c_{\pmb \epsilon}(k))_{k\in \Z}$ by
\begin{equation}  \label{phaselessminimization.4}
 c_{\pmb\epsilon}(k)={c}^2_{\pmb\epsilon, k'}(k)\end{equation}
\end{itemize}
where $ k'= \lfloor(2k+L-1)/(2L)\rfloor, k\in \Z$.

The above algorithm contains
four steps:
1) solving the {\bf M}inimization problem \eqref{phaselessminimization.1} to obtain local approximations
${\bf c}^0_{\pmb\epsilon,k'}, k'\in \Z$, of ${\bf c}$ on $k'L+[-N+1, 0]$, up to a phase $\delta_{k'}\in \{-1, 1\}$;
 2)  {\bf  E}xtending ${\bf c}^0_{\pmb\epsilon,k'}$  to new local approximations ${\bf c}^1_{\pmb\epsilon,k'}, k'\in \Z$,
  of $\delta_{k'}{\bf c}$ on $k'L+[1-N-(L-1)/2, (L-1)/2]$; 
3) 
adjusting {\bf P}hases of ${\bf c}^1_{\pmb \epsilon, k'}$ to obtain local approximations ${\bf c}^2_{\pmb\epsilon,k'}$
to either ${\bf c}$ or $-{\bf c}$ on  $k'L+[1-N-(L-1)/2, (L-1)/2]$; 
and 4) {\bf S}ewing ${\bf c}^2_{\pmb\epsilon,k'}, k'\in \Z$, together to get the approximation
${\bf c}_{\pmb\epsilon}$ to  either ${\bf c}$ or $-{\bf c}$. We call the algorithm
\eqref{c.epsilon.0.def}--\eqref{phaselessminimization.4}
 as the MEPS algorithm.

In  the noiseless sampling environment (i.e., ${\pmb \epsilon}={\bf 0}$),  we can set the threshold  $M_0=0$. Then
there exist signs $\delta_{k'}\in \{-1, 1\}, k'\in \Z$, and
$\delta\in \{-1, 1\}$ such that
$$
c_{{\bf 0}, k'}^0(k)=\delta_{k'} c(k),\ \  k\in k'L+[- N + 1, 0];
$$
$$c^1_{{\bf 0},k'}(k)= \delta_{k'} c(k),\ \  k\in k'L+\big[- N + 1-\frac{L-1}{2}, \frac{L-1}{2}\big];$$
and
 \begin{equation}  \label{phaselessminimization.4new}
 {c}_{{\bf 0}}^2(k)=\delta c(k), \ \  k\in \Z.\end{equation}
Therefore  the MEPS algorithm  provides a perfect reconstruction of  a nonseparable signal, up to a sign,   in the noiseless sampling environment.

For  the nonseparable signal $f\in V(\phi)$  in \eqref{signal.def}, set
\begin{equation} \label{mf.def} M_f= \sup_{K_-(f)-N+1<k<K_+(f)+1}
\frac{\sum_{l=-1}^{N-1} |c(k+l)|^2}{\sum_{l=0}^{N-2} |c(k+l)|^2}.
\end{equation}
 In the next theorem, we show that the MEPS algorithm \eqref{c.epsilon.0.def}--\eqref{phaselessminimization.4} provides, up to a sign,   a stable  approximation to the original
nonseparable signal $f$ in a noisy sampling environment.

 \begin{theorem}\label{realother.thm}
 Let $\phi$, $X$, $\Gamma$ and $\Gamma^\ast$  be as in Theorem \ref{realcase.thm}, 
$f(t)=\sum_{k=-\infty}^\infty c(k) \phi(t-k)$ in \eqref{signal.def} be a nonseparable real-valued signal
with
$S_f$ in \eqref{sf.def} being positive and $M_f$ in \eqref{mf.def} being finite,
and let  $f_{\pmb \epsilon}(t)=\sum_{k\in \Z} c_{\pmb \epsilon}(k)\phi(t-k)$  be the signal in \eqref{fepsilon.def} reconstructed
by the  MEPS algorithm \eqref{c.epsilon.0.def}--\eqref{phaselessminimization.4} with the threshold
\begin{equation}\label{M0.def}
M_0= \frac{S_f}{4\|(\Phi_N)^{-1}\|^{2}}.
\end{equation}
 If
\begin{equation}\label{realother.thm.eq2}
|{\pmb \epsilon}| \le
\frac{ S_f}
{2^7 N^3  \|(\Phi_N)^{-1}\|^2 (C_{f, \phi})^{4N+L-5}}, \end{equation}
 then
 there exists $\delta\in \{-1, 1\}$ such that
\begin{equation} \label{realother.thm.eq3}
 |c_{\pmb \epsilon}(k)-\delta c(k)|\le
 N\|(\Phi_N)^{-1}\| (C_{f, \phi})^{N-1+(L-1)/2} \sqrt{8|\pmb\epsilon|}
\end{equation}
for all $k\in \Z$, where
\begin{equation}\label{cfphi.def}
C_{f, \phi}=
  \frac{ 2^{8}   \|\Phi\|^4\|(\Phi_N)^{-1}\|^3 M_f }
{\min_{1\le n\le N} \{|\phi(\gamma_n)|, |\phi(\gamma_n^*+N-1)|\} }.
\end{equation}
 \end{theorem}

The  requirement
\eqref{realother.thm.eq2}   on the noise level $|{\pmb \epsilon}|$ has exponential decay  about $L\ge 1$.
Our numerical simulations in the next section indicate that  for large $L$, the MEPS algorithm may fail to save phase of a nonseparable signal
(and hence reconstruct the signal approximately)
in a noisy sampling environment.

\section{Numerical simulations}
\label{simulation.section}

 In this section, we demonstrate the performance of the MEPS algorithm  \eqref{c.epsilon.0.def}--\eqref{phaselessminimization.4} to reconstruct
a nonseparable cubic spline signal
\begin{equation} \label{signal.sim}
f(t)=\sum_{k\in \Z} c(k) B_4(t-k)\end{equation}
with finite duration, where   $B_4$ is the cubic B-spline in \eqref{bspline.def}.
Our noisy phaseless samples are taken on $Y_L$,
\begin{equation}\label{data.sim}
z_{\pmb\epsilon}(y )=|f(y)|^2+{\pmb \epsilon}(y)*\|f\|_\infty^2\ge 0, \  y\in Y_L,\end{equation}
 where $L$ is an odd integer,   ${\pmb\epsilon}(y)\in [-\varepsilon, \varepsilon], y\in Y_L$, are randomly selected with noise level $\varepsilon>0$, and $$Y_L=\Big ( \Big\{\frac{m}{8}, 1\le m\le 7\Big\}+L\Z\Big)\bigcup\Big(\Big\{\frac{1}{8},\frac{3}{8}, \frac{5}{8}, \frac{7}{8}\Big\}+\Z\Big)
$$
 has sampling rate $4+3/L$.
In our simulations,
\begin{equation}\label{ck.sim}c(k)\in  [-1,1]\setminus[-0.1,0.1],\ K_-(f)\le k\le K_+(f), \end{equation} 
are randomly selected. 
Denote  the signal reconstructed by the MEPS algorithm from the noisy phaseless samples
\eqref{data.sim} by
\begin{equation} \label{fepsilon_dem}
f_{\varepsilon, L}(t)=\sum_{k\in \Z} c_{\varepsilon, L}(k) B_4(t-k).\end{equation}
Shown in  Figure \ref{splinesignal} are a nonseparable cubic spline signal  $f$ 
and  the  reconstruction error  $f_{\varepsilon, L}-f$, 
which demonstrates the stability of the MEPS algorithm for phase retrieval of nonseparable cubic spline signals.
\begin{figure}[h]
\center
\begin{tabular}{cc}
\includegraphics[width=42mm, height=35mm]{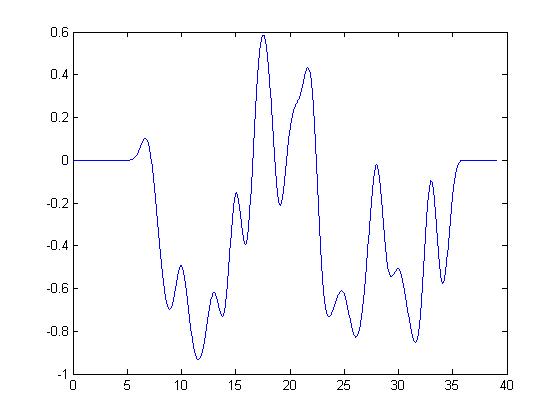} &
\includegraphics[width=42mm, height=35mm]{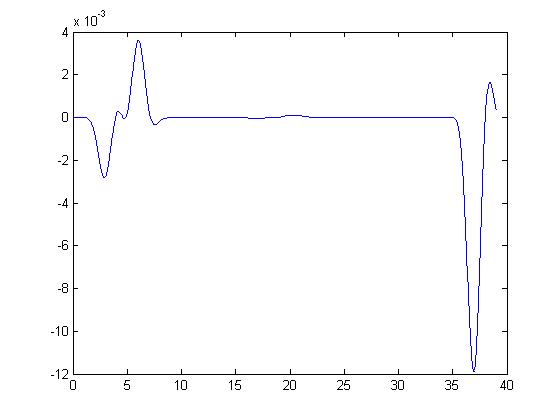} 
\end{tabular}
\caption{Plotted on the left is a nonseparable cubic spline signal $f$ with
 $K_-(f)=5,  K_+(f)=32$ and $c(k), k\in \Z$, in \eqref{ck.sim}.
On the  right
is the difference between the above signal $f$   and the signal $f_{\varepsilon, L}$
reconstructed by the MEPS algorithm
from the noisy samples
\eqref{data.sim} with $\varepsilon=10^{-7}$ and $L=7$.  It is observed that the
 MEPS algorithm has higher accuracy to recover a signal  inside its support.
}
\label{splinesignal}
\end{figure}

Define a maximal reconstruction error of the MEPS algorithm by
  \begin{equation}\label{maximalerror.def}
  e(\varepsilon,L):=\min_{\delta\in \{-1, 1\}}
 \max_{k\in \Z} |c_{\varepsilon, L}(k)-\delta c(k)|.\end{equation}
As the cubic B-spline $B_4$ is a nonnegative function satisfying 
$$\sum_{k\in \Z} B_4(t-k)=1\ {\rm for \ all}\  t\in \R,$$
we have
$$\min_{\delta\in \{-1, 1\}}
 \max_{t\in \R}|f_{\varepsilon, L}(t)-\delta f(t)|\le  e(\varepsilon,L).$$ 
For large odd integers $L$, the MEPS algorithm may not yield an approximation to the original signal in a noisy environment, as in Theorem \ref{realother.thm}
the stability requirement \eqref{realother.thm.eq2}  on the noise level $\varepsilon$ has exponential decay about $L\ge 1$.
 Our numerical simulations show that
for large  odd $L$, the MEPS algorithm  may fail to save phases of nonseparable cubic spline signals, but
its success rate  to save phases (and then to reconstruct signals approximately)  is still high  for large $L$. Presented in
Table \ref{iterationerror_revised.tab} is the success rate  after 500 trials
for different noisy levels $\varepsilon$ and extension lengths $L$ to recover  cubic spline signals $f$ in \eqref{signal.sim}
with
$c(k), k\in \Z$, in \eqref{ck.sim} and noisy samples
in \eqref{data.sim}.
 Here the MEPS algorithm is considered to save the phase successfully if
 \begin{equation}\label{nonseperable.eq}
  e(\varepsilon,L)< 0.1. \end{equation}
  \begin{table}
\caption{\normalsize Success rate 
for different noisy levels $\varepsilon$ and extension lengths $L$.}
\centering
\begin{tabular}{|c|c|c|c|c|c|c|c|}
\hline
\hline
  \backslashbox {$\varepsilon$}{$L$}
 & 7
&11 &15&23 &31 &47  \\
\hline
 $10^{-5}$& $0.3140$ & $0.2400$ &$0.1180$&$0.0500$&0.0180& $0.0040$ \\
\hline
 $10^{-6}$& $0.8440$  & $0.7780$ &0.7360&0.5980&0.5200& $0.4220$ \\
\hline
 $10^{-7}$& $0.9840$  & $0.9760$ &0.9660&0.9480&0.9340&  $0.9020$\\
\hline
 $10^{-8}$& $0.9980$ & $0.9960$ &1&0.9980& $0.9860$&0.9980\\
\hline
 $\le 10^{-9}$& $1$  & $1$ &1&1&1& 1\\
\hline
\end{tabular}
\label{iterationerror_revised.tab}
\end{table}
In the simulation,  a successful recovery 
implies that
$c_{\epsilon, L}(k)$ and $c(k), K_-(f)\le k\le K_+(f)$,  have same signs,
$$ c_{\varepsilon, L}(k) c(k)>0 \ {\rm for \ all} \ K_-(f)\le k\le K_+(f).$$
 The threshold  selected in \eqref{nonseperable.eq} for the maximal  reconstruction error $e(\varepsilon, L)$
 is less than
 $$\min_{K_-(f)\le k\le K_+(f)-N+2}\Big(\frac{1}{N-1} \sum_{l=0}^{N-2} |c(k+l)|^2\Big)^{1/2},$$
 which is similar to the quantity $S_f$ in \eqref{sf.def} to measure the distance of a nonseparable signal $f$ to the set of all separable signals in a shift-invariant space, cf. Theorem \ref{realX.thm}.

By \eqref{c.half.leastsquare}, \eqref{III.11} and Theorem \ref{realother.thm},
the maximal reconstruction error $e(\varepsilon, L)$ in  \eqref{maximalerror.def}
and  the reconstruction errors $|c_{\varepsilon, L}(k)-c(k)|, k\not\in [K_-(f), K_+(f)]$, outside the support region are about the order $\sqrt{\varepsilon}$.
Numerical simulations indicate that
the reconstruction errors $|c_{\varepsilon, L}(k)-c(k)|, K_-(f)+N\le k\le K_+(f)-N$,  are about the order $\varepsilon$, which is much smaller than
the maximal reconstruction error $e(\varepsilon, L)$, see Figure \ref{revised.fig}.
\begin{figure}[h]
\center
\begin{tabular}{cc}
\includegraphics[width=42mm, height=35mm]{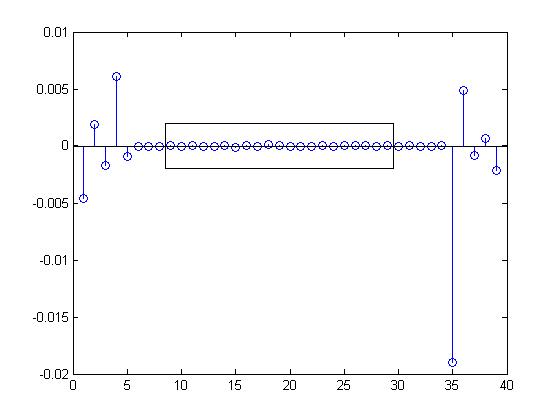} &
\includegraphics[width=42mm, height=35mm]{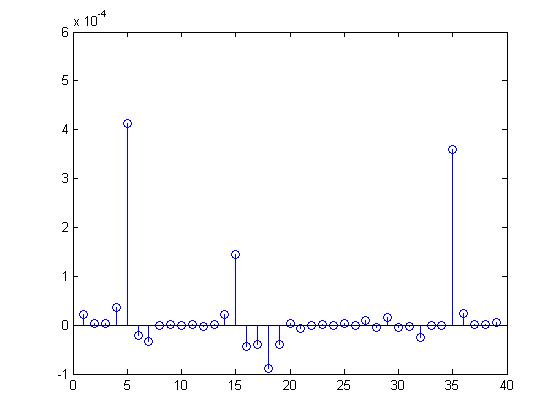}\\
\end{tabular}
\caption{Plotted on the left is the difference $c_{\varepsilon, L}(k)-c(k), 1\le k\le 39,$
between the reconstructed amplitudes 
 and original amplitudes in Figure 1, 
 while on the right
 is the squared difference $(c_{\varepsilon, L}(k))^2-(c(k))^2, 1\le k\le 39$. In the simulation, the maximal reconstruction error $e(\varepsilon, L)$ is $0.0190$, the maximal squared reconstruction error $e_2(\varepsilon, L)$ in \eqref{maximalerror2.def} is $4.1204\times 10^{-4}$,
 and the maximal reconstruction error $\max_{K_-(f)+N\le k\le K_+(f)-N} |c_{\varepsilon, L}(k)-c(k)|$
  inside the box
  is $1.5041\times 10^{-4}$. The maximal squared reconstructed error and the maximal reconstructed error inside the box are at the same order of
        $\|(\pmb \Phi_N)^{-1}\|\epsilon\approx 1.2390\times 10^{-4}$, cf. \eqref{realother.thm.eq3}.
 }
\label{revised.fig}
\end{figure}

An alternative to  measure
the phase retrieval error is the following maximal squared reconstructed error,
  \begin{equation}\label{maximalerror2.def}
  e_2(\varepsilon,L):=
 \max_{k\in \Z} |(c_{\varepsilon, L}(k))^2- (c(k))^2|.\end{equation}
For small $\varepsilon>0$, it follows from Theorem \ref{realother.thm}  that
$$e_2(\varepsilon, L)\le 3
e(\varepsilon, L) \big(\max_{k\in \Z}|c(k)|\big)
\le C \big(\max_{k\in \Z}|c(k)|\big)\sqrt{\varepsilon},$$
where $C$ is a positive constant. The above upper bound estimate  for the measurement
$e_2(\varepsilon, L)$  should not be optimal, as our numerical simulations indicate that the above alternative measurement  $e_2(\varepsilon,L)$ is about the order $\varepsilon$, see
Figure \ref{revised.fig}.

The success rate of the MEPS algorithm could have
 significant improvement  if the phaseless samples
$|f(y)|, y\in Y_L\cap [K_-(f), K_+(f)]$, of the original signal $f$ in \eqref{signal.sim} are a distance away from the origin. Such a requirement holds
if the 
cubic spline signal $f$ has only ``one" phase,  i.e.,
$c(k)>0$ for all $K_-(f)\le k\le K_+(f)$.
 Presented in Table \ref{iterationerror_positive.tab}
is the success rate of the MEPS algorithm to recover the positive phase of nonseparable  cubic spline signals $f$ in \eqref{signal.sim}
with
\begin{equation}\label{ckpositive.sim}c(k)\in  [0.1,1] \ {\rm for \ all} \  K_-(f)\le k\le K_+(f) \end{equation}
after 500 trails, where the noise level $\varepsilon$, the extension length $L$ and
the success threshold are the same as in Table  \ref{iterationerror_revised.tab}.
 \begin{table}
\caption{Success rate 
 to recover cubic spline 
 signals with ``one" phase}
\centering
\begin{tabular}{|c|c|c|c|c|c|c|}
\hline
\hline
  \backslashbox {$\delta$}{L}
 & 7
&11 &15&23 &31 &47 \\
\hline
 $10^{-4}$& 0.6640  & $0.6300$ &0.5740&0.5420&0.5000& 0.4040\\
\hline
 $5*10^{-5}$& 0.8840 & $0.9080$ &0.8600&0.8600&0.8540&  0.8520\\
\hline
 $\le 10^{-5}$& 1 & $1$ &1&1&1& 1\\
\hline
\hline
\end{tabular}
\label{iterationerror_positive.tab}
\end{table}
Under the ``one" phase assumption on the original signal, the extension parts (ii) and (iii) in the MEPS algorithm
does not propagate noises at each extension step, because
 for all $1\le l\le (L-1)/2$,
signs  of
$\tilde z_{\epsilon}(k'L+l+\gamma_n)$ in \eqref{newhalf.eq3} are the same as $f( k'L+l+\gamma_n), 1\le l\le (L-1)/2$, and
similarly for $1\le l'\le (L-1)/2$,  signs  of
$\tilde z_{\epsilon}(k'L-l'+\gamma_n^*)$ in \eqref{newone.eq3-} are the same as $f( k'L-l+\gamma_n^*), 1\le n\le N$. 

\section{Conclusions} 

Let ${\mathcal S}(\phi)$ be the set of all real-valued signals in a shift-invariant space $V(\phi)$
that can, up to a sign,  be reconstructed from its  magnitude on the whole line.
For a compactly supported continuous generator $\phi$,
${\mathcal S}(\phi)$ is neither the whole linear space $V(\phi)$ nor its convex subset.
 This 
 is a different phenomenon from  the bandlimited case,  for which it is observed that all bandlimited signals  can, up to a sign, be reconstructed from its  magnitude on the whole line (\cite{T11, shenoy16}).

Phase retrieval of signals in a shift-invariant space is a sampling and reconstruction problem.
The set ${\mathcal S}(\phi)$ contains all nonseparable signals, which could be
determined from its phaseless sampling on some sets with sampling rate large than the support length of  the generator $\phi$.

Many algorithms have been introduced to solve a phase
retrieval problem in the finite-dimensional setting.
 The MEPS algorithm is proposed to solve the infinite-dimensional
 phase retrieval problem for nonseparable signals in  a shift-invariant space.
  The MEPS algorithm can be implemented in a distributed manner (\cite{bertsekasbook1989, chengsds}), and it is stable against bounded sampling noises.

\begin{appendices}
\section{Proof of Theorem \ref{separable.tm}}
($\Longrightarrow$) \ Suppose, on the contrary, that there exist nonzero signals $f_1, f_2\in V$ such that $f=f_1+f_2$
and $f_1f_2=0$.  Set $g=f_1-f_2\in V$. Then $g\ne \pm f$ and  $|g|=|f|=\sqrt{|f_1|^2+|f_2|^2}$. This is a contradiction.

($\Longleftarrow$) \ Let $g$ be a signal in $V$ with $|g|=|f|$.  Set
$ g_1:=(f+g)/2\in V$ and $g_2:=(f-g)/2\in V$.
Then $f=g_1+g_2$ and $g_1g_2=0$. This together with  nonseparability of the signal $f$ implies that
    either $g_1 \equiv 0$ or $g_2\equiv 0$. Hence 
    $g$ is either $-f$ or $f$. This completes the proof. 

\section{Proof of Theorem \ref{realX.thm}}
 We divide the proof into three implications iii)$\Longrightarrow$i), i)$\Longrightarrow$ii) and ii)$\Longrightarrow$iii).


iii)$\Longrightarrow$i): \ The implication follows immediately from Theorem \ref{separable.tm}.

i)$\Longrightarrow$ii): \
Set  $K_{\pm}=K_{\pm}(f)$. 
For $K_-+1-N<k<K_-+1$ or $K_++1-N<k< K_++1$, the conclusion $\sum_{l=0}^{N-2}|c(k+l)|^2\ne 0$  follows from  the definitions of $K_-$ and $K_+$.
Then it remains to  establish the statement ii) for $K_-< k< K_++2-N$. 
Suppose, on the contrary, that
\begin{equation}\label{zeroassumption} \sum_{l=0}^{N-2}|c(k_1+l)|^2=0\end{equation} 
for some $K_-< k_1< K_+-N+2$.
 Set  \begin{equation*} 
 f_1(t):=\sum_{l=K_-}^{k_1-1}c(l)\phi(t-l) \ \  {\rm and} \ \  f_2(t):=\sum_{l=k_1+N-1}^{K_+}c(l)\phi(t-l). \end{equation*}
Then
 \begin{equation} \label{ff1f2}
  f=f_1+f_2\ \ {\rm and} \ \ f_1f_2=0 \end{equation}
  by \eqref{zeroassumption} and
the observation that $f_1$ and $f_2$ are  supported in $(-\infty, k_1+ N-1]$ and $[k_1+N-1,\infty)$ respectively.
 Clearly, $f_1$ and $f_2$ are nonzero signals in $V(\phi)$. This together with
  \eqref{ff1f2} 
  implies that $f$ is separable, which  contradicts to the assumption i).

ii)$\Longrightarrow$iii): \ To prove the implication, we need a lemma.
\begin{lemma}\label{phaseretrievalL.lem}
Let $\phi$ and $X$ be  as in Theorem \ref{realX.thm}.
 Then   for any $l\in \Z$ and  signal $g(t)=\sum_{k=-\infty}^{\infty}d(k)\phi(t-k)\in V(\phi)$,
  coefficients
$d(k), l-N+1\le k\le l$, are completely determined, up to a sign, by phaseless samples $|g(x_m+l)|, \ x_m\in X$, of the signal $g$.
\end{lemma}
The above lemma   follows immediately from \cite[Theorem 2.8]{BCE06} and the observation that
 $$g(x_m+l)=\sum_{k=l-N+1}^{l}d(k)\phi(x_m+l-k), \ \ x_m\in X.$$

Take a particular integer $K_--1<k_0<K_++1$ with $c(k_0)\ne 0$. Without loss of generality, we assume that
\begin{equation}\label{ko.assumption1} c(k_0)>0,\end{equation}
otherwise replacing $f$ by $-f$.

Using \eqref{ko.assumption1} and applying Lemma \ref{phaseretrievalL.lem} with $g$  and $l$ replaced by $f$ and  $k_0$ respectively, we conclude that
$c({k_0-N+1}), \cdots, c({k_0})$ are completely determined by phaseless samples $|f(X+k_0)|$ of the signal $f$ on $X+k_0$.
Now we prove that
\begin{equation}\label{B5.eq}
c(k), \ k\le k_0,\ {\rm  are\  determined\ by} \  |f(X+k)|,\ k\le k_0\end{equation}
by induction.
Inductively we assume that  $c({k}), k_0-p-N+1\le k\le k_0$, are determined from $|f(X+k)|, k_0-p\le  k\le k_0$.
The inductive proof is complete if $k_0-p-N+1\le K_-$. Otherwise $k_0-p-N+1>K_-$ and
\begin{equation}\label{inductive.hyp} \sum_{l=0}^{N-2} |c({k_0-p-N+l+1})|^2\ne 0\end{equation}
 by the assumption ii). Applying Lemma \ref{phaseretrievalL.lem} with $g$
and $k_0$ replaced by $f$ and  $k_0-p-1$ respectively, we conclude that $c({k_0-N-p}), \cdots, c({k_0-p-1})$ are determined by
$|f(X+k_0-p-1)|$ up to a global phase.  This together with \eqref{inductive.hyp} and the inductive hypothesis
implies that $c({k_0-N-p}), \cdots, c({k_0-p-1})$  is determined by
$|f(X+k)|, k_0-p-1\le k\le k_0$. Thus the inductive argument can proceed.

Using the similar argument, we can show that
\begin{equation}\label{B7.eq}c(k),\ k\ge  k_0,\ {\rm  are\  determined\ by} \  |f(X+k)|, k\ge k_0.\end{equation}
 Combining \eqref{B5.eq} and \eqref{B7.eq} completes the proof.

\section{Proof of Theorem \ref{necessary.thm}}
By \eqref{samplingrate.def}  it suffices to prove that
\begin{equation*}\label{necessary.pf.eq1} \#(I\cap [a, b])\ge b-a-N+1\end{equation*}
for all integers $a$ and $b$ with $b-a\ge N$. Suppose, on the contrary, that
\begin{equation}\label{necessary.pf.eq2} \#(I\cap [a_0, b_0])<b_0-a_0-N+1\end{equation}
for some integers $a_0$ and  $b_0$. Let
\begin{equation*}
{\mathcal N}=\Big\{ f(t):=\sum_{k=a_0}^{b_0-N} c(k) \phi(t-k), \ f(y)=0 \ {\rm for \ all}\ y\in I\Big\}.
\end{equation*}
Then ${\mathcal N}$ contains some nonzero signals in $V(\phi)$, because
any signal of the form $\sum_{k=a_0}^{b_0-N} c(k) \phi(t-k)$ is supported in $[a_0, b_0]$, and
the homogenous linear system
$$\sum_{k=a_0}^{b_0-N} c(k) \phi(y-k)=0, \ \ y\in I\cap [a_0, b_0]$$
of size $(\#(I\cap [a_0, b_0]))\times (b_0-a_0-N+1)$ has a nontrivial solution by \eqref{necessary.pf.eq2}.

Take a nonzero signal $f\in {\mathcal N}$ with minimal support length.   By the assumption on the set $I$, it must be separable
as it is a nonzero signal having zero  magnitude measurements on $I$. Therefore by  Theorem \ref{realX.thm}
there exist nonzero signals $f_1$ and $f_2\in V(\phi)$ and an integer $k_0\in (a_0, b_0)$ such that $f_1$ vanishes outside $[k_0, b_0]$, $f_2$  vanishes outside $[a_0, k_0]$ and
$ f=f_1+f_2$.  This implies that both $f_1$ and $f_2$ are nonzero  signals in ${\mathcal N}$, which contradicts to the assumption that
$f\in {\mathcal N}$ has minimal support length.

\section{Proof of Theorem \ref{realcase.thm}}
To prove Theorem \ref{realcase.thm}, we need a technical lemma.

\begin{lemma}\label{realcase.lem}
Let $\gamma_n$ and $\gamma_n^\ast$, $ 1\le n\le N$, and $\phi$ be  as in Theorem \ref{realcase.thm}. Then
\begin{equation}\label{realcase.lem.eq1}
\left(\begin{array}{ccc} \phi(\gamma_1)  & \ldots & \phi(\gamma_N)\\
 \displaystyle{\sum_{l=0}^{N-2}}  a(l) \phi(\gamma_{1}+l+1) & \ldots &  \displaystyle{\sum_{l=0}^{N-2}}  a(l) \phi(\gamma_{N}+l+1)\end{array}\right)
\end{equation}
and
\begin{equation}\label{realcase.lem.eq2}
\left(\begin{array}{ccc} \phi(\gamma^\ast_1+N-1)  & \ldots & \phi(\gamma^\ast_N+N-1)\\
 \displaystyle{\sum_{l=0}^{N-2}}  a(l) \phi(\gamma^\ast_{1}+l) & \ldots &  \displaystyle{\sum_{l=0}^{N-2}}  a(l) \phi(\gamma^\ast_{N}+l)\end{array}\right)
\end{equation}
have rank $2$ for any nonzero vector $(a(0), \ldots, a(N-2))$. 
 \end{lemma}
 
 \begin{proof}
 We prove \eqref{realcase.lem.eq1} by indirect proof.
Suppose, on the contrary, that
$$\frac{\sum_{l=0}^{N-2}  a(l) \phi(\gamma_{n}+l+1)} {\phi(\gamma_{n})}=\alpha, \ \ 1\le n\le N$$
for some  $\alpha\in \R$.
This together with nonsingularity  of the matrix $(\phi(\gamma_n+m))_{1\le n\le N, 0\le m\le N-1}$ implies that
$(-\alpha, a(0), \ldots, a({N-2}))$ is a zero vector, which is a contradiction.

The full rank property \eqref{realcase.lem.eq2} can be proved similarly. 
 \end{proof}

\begin{proof}[Proof of Theorem \ref{realcase.thm}]   Due to the shift-invariance, without loss of generality, we assume that $k_0=0$.
Set  $K_{\pm}=K_{\pm}(f)$.
We divide the proof into three cases: $K_-\le 0\le K_++N-1$, $K_-\ge 1$ and $K_+\le -N$.

 {\em Case 1: $K_-\le 0\le K_++N-1$.}

 In this case, it follows from Theorem \ref{realX.thm} and  nonseparability of the signal $f$ that
  there exists  $-N+2\le n_0\le 0$ such  that $c({l_0})\ne 0$ and $c(l)=0$ for all $l_0<l \le 0$.
  Without loss of generality, we assume that
\begin{equation}\label{ko.assumption} c({l_0})>0 \ {\ \rm and \ }\  c(l)=0  \ {\ \rm for \ all \ } \ l_0< l\le 0,\end{equation}
otherwise replacing $f$ by $-f$.
 By  \eqref{ko.assumption} and Lemma \ref{phaseretrievalL.lem},
 \begin{equation}\label{ckvalue.def}
 c({-N+1}),\ldots,c(-1), c(0)\end{equation}
 are  determined  from phaseless samples $|f(X)|$. 
Next we prove by induction that  $c(k), k\ge -N+1$, are determined by  phaseless samples $|f(X)|$ and $|f(\Gamma+q)|, q\ge 1$.
Inductively, we assume that $c(k), -N+1\le k<p$, can be recovered from  $|f(X)|$ and  $|f(\Gamma+q)|, 1\le q<p$. The induction proof is finished if $p>K_+$. Now it remains to consider $p\le K_+$.

 Observe that
 \begin{eqnarray}\label{signalLplus1.def}
 f(\gamma_n+p)&\hskip-0.08in =&  \hskip-0.08in c(p)\phi(\gamma_n)+\sum_{l=0}^{N-2}c({p-l-1})\phi(\gamma_n+l+1)\nonumber\\
 & \hskip-0.08in =:& \hskip-0.08in c(p)\phi(\gamma_n)+\alpha(n)\ \ {\rm for \ all} \  \gamma_n\in \Gamma.
 \end{eqnarray}
 Taking squares at both sides of the above equations yields
$$
|\phi(\gamma_n)|^2 (c(p))^2+2 \phi(\gamma_n) \alpha(n) c(p)+|\alpha(n)|^2= |f(\gamma_n+p)|^2,
$$
 where $1\le n\le N$. Moving $|\alpha(n)|^2$ to the right hand side and then dividing $\phi(\gamma_n)$ at both sides, we obtain
 \begin{equation} \label{signalLplus1.def2}
\phi(\gamma_n) (c(p))^2+2  \alpha(n) c(p)= \frac{|f(\gamma_n+p)|^2-|\alpha(n)|^2}{\phi(\gamma_n)},
 \end{equation}
where $1\le n\le N$.
As $K_-< p\le K_+$, we obtain from Theorem \ref{realX.thm} that $(c({p-N+1}), \ldots, c({p-1}))$ is a nonzero vector. Therefore by Lemma \ref{realcase.lem},
the $2\times N$ matrix
$$\left(\begin{array}{ccc} \phi(\gamma_1)  & \cdots & \phi(\gamma_N)\\
\alpha(1) & \ldots &  \alpha(N)\end{array}\right)$$
has rank $2$. So there is a unique 
 solution
\begin{equation}\label{cp.def}
c(p)= \frac
{h_2({\pmb \alpha}, {\pmb \eta})}
{  2 h_1({\pmb \alpha}) }
\end{equation}
to the linear system \eqref{signalLplus1.def2}, where $h_1, h_2$ are functions given in \eqref{gGamma.def} and \eqref{gGamma2.def} respectively, ${\pmb \alpha}=(\alpha(1), \ldots, \alpha(N))$,
and
$${\pmb \eta}=(|f(\gamma_1+p)|^2-|\alpha(1)|^2, \ldots, |f(\gamma_N+p)|^2-|\alpha(N)|^2).$$
  This completes the inductive proof. Hence
$c(k), k\ge -N+1$,  are determined  from  $|f(X)|$ and  $|f(\Gamma+q)|, q\ge 1$.

\smallskip

 Finally we use similar arguments to determine $c(k),  k\le -N+1$, from $|f(X)|$ and  $|f(\Gamma^{\ast}+q)|, q\le -1$.
  Inductively, we assume that $c(k), \tilde p<k\le 0$, has been recovered from   $|f(X)|$ and  $|f(\Gamma^{\ast}+q)|, \tilde p+N-1<q\le -1$. The induction proof is done if $\tilde p<K_-$. Then it remains to discuss $\tilde p\ge K_-$.
Observe that
{
 \begin{eqnarray}\label{signalL1.def}
 f(\gamma^{\ast\ast}_n+\tilde p) &\hskip-0.1in = &\hskip-0.1in  c({\tilde p})\phi(\gamma^{\ast\ast}_n) +  
  \sum_{l=0}^{N-2}c(\tilde p+N-1-l)\phi(\gamma^*_n+l)\nonumber\\
 &\hskip-0.1in =: &\hskip-0.1in  c({\tilde p})\phi(\gamma_n^{\ast\ast})+\alpha^\ast(n).
 \end{eqnarray}}
 where $\gamma_n^{\ast\ast}=\gamma^\ast_n+N-1, 1\le n\le N$.
By Lemma \ref{realcase.lem}, the $2\times N$ matrix
\begin{equation*}
\left(\begin{array}{ccc} \phi(\gamma_1^{\ast\ast})  & \ldots & \phi(\gamma^{\ast\ast}_N)\\
 \alpha^\ast(1) & \ldots &  \alpha^\ast(N)
 \end{array}\right)
\end{equation*}
has rank $2$. Therefore
\begin{equation}\label{ctildeL.formula}
c_{\tilde p}=\frac
{h_2^*({\pmb \alpha}^*, {\pmb \eta}^*)
}
{  2 h_1^*({\pmb \alpha}^*)},
\end{equation}
where $h_1^*$ and $h_2^*$ are 
 given in \eqref{gGammastar.def} and \eqref{gGamma2star.def} respectively, ${\pmb \alpha}^\ast=(\alpha^\ast(1), \ldots, \alpha^\ast (N))$,
and
$${\pmb \eta}^\ast=(|f(\gamma_1^{\ast\ast}+\tilde p)|^2-|\alpha^{\ast}(1)|^2, \ldots, |f(\gamma_N^{\ast\ast}+\tilde p)|^2-|\alpha^\ast(N)|^2).$$
This completes the inductive proof. Therefore
 $c(k), k\le 0$, are determined  from  $|f(X)|$ and  $|f(\Gamma^{\ast}+q)|, q\le -1$.

{\em Case 2: $K_-\ge 1$.}

In this case, the signal $f$ is supported in $[1, \infty)$. Without loss of generality, we assume  that  $c({K_-})>0$, otherwise considering $-f$ instead of $f$.  From
the definition of $K_-$ and the supporting property of  $\phi$, we have
\begin{equation*}
f(\gamma_n+K_-)=c({K_-})\phi(\gamma_n), \ \ \gamma_n\in \Gamma.
\end{equation*}
Thus
$$c(K_-)=\frac{\sum_{n=1}^N |\phi(\gamma_n)| |f(\gamma_n+K_-)|}{\sum_{n=1}^N |\phi(\gamma_n)|^2}.$$
 Then  following the same procedure as in Case 1,  we obtain that $c(k), k\ge K_-$, are determined from
 $|f(\Gamma+q)|, q\ge K_-$.

{\em Case 3: $K_{+}\le -N$.}

In this case, the signal $f$ is supported in $(-\infty, 0]$, and
$c({K_+})$ can be obtained,  up to a sign, from phaseless samples $|f(\Gamma^\ast+K_++N-1)|$.
Following the same procedure as in Case 1, we can determine  $c(k), k\le K_+\le -N$,  from  $|f(\Gamma^\ast+q)|, q\le
K_++N-1$.
\end{proof}

\section{Proof of Theorem \ref{realother.thm}}

The proof of Theorem \ref{realother.thm} is quite technical.
 It includes  three propositions on the approximation property of  vectors in the first three steps of the MEPS algorithm  \eqref{c.epsilon.0.def}--\eqref{phaselessminimization.4},
and one proposition on the phase adjustment.

To prove Theorem \ref{realother.thm}, we first show that for any $k'\in \Z$,  the  vector  ${\bf c}_{\pmb \epsilon, k'}^0$ in
 the first step of  the MEPS algorithm
approximates, up to a sign,  the original vector ${\bf c}$ on $[k'L+1-N, k'L]$.

\begin{proposition}\label{c0.pr}
Let ${\bf c}, {\pmb \epsilon}$ be as in Theorem \ref{realother.thm}, and ${\bf c}_{\pmb \epsilon, k'}^0$, $k'\in\Z$, be as in \eqref{c.epsilon.0.def}. Then for any $k'\in \Z$, there exists $\delta_{k'}\in \{-1, 1\}$ such that
\begin{equation}\label{realother.thm.pf.eq3}
\sum_{k=k'L-N+1}^{k'L} |c_{\pmb \epsilon, k'}^0 (k)-\delta_{k'} c(k)|^2\le  8N \|(\Phi_N)^{-1}\|^2 |{\pmb \epsilon}|.
\end{equation}
\end{proposition}

\begin{proof} Set  $x_{m,k'}=x_m+k'L, 1\le m\le 2N-1$.
Then
\begin{eqnarray*}
&\hskip-0.08in &  \hskip-0.08in\sum_{m=1}^{2N-1} \Big(\Big|\sum_{k=k'L-N+1}^{k'L} c_{\pmb \epsilon, k'}^0 (k) \phi(x_{m,k'}-k)\Big|\nonumber\\
& \hskip-0.08in & \hskip-0.08in\qquad \qquad - \Big|\sum_{k=k'L-N+1}^{k'L} c (k) \phi(x_{m,k'}-k)\Big|\Big)^2\nonumber\\
&\hskip-0.08in \le  & \hskip-0.08in 2 \sum_{m=1}^{2N-1} \Bigg(\Big|\sum_{k=k'L-N+1}^{k'L} c_{\pmb \epsilon, k'}^0 (k) \phi(x_{m,k'}-k)\Big|\nonumber\\
 & \hskip-0.08in & \hskip-0.08in \quad  - \sqrt{z_{\pmb\epsilon}(x_{m,k'})}\Bigg )^2  +2 \sum_{m=0}^{2N-1} \Bigg(  \sqrt{z_{\pmb \epsilon}(x_{m,k'})}\nonumber\\
& \hskip-0.08in & \hskip-0.08in \qquad \qquad - \Big|\sum_{k=k'L-N+1}^{k'L} c(k) \phi(x_{m,k'}-k)\Big|\Bigg)^2\nonumber\\
& \hskip-0.08in \le & \hskip-0.08in  4  \sum_{m=1}^{2N-1} \Big|| f(x_{m,k'})| - \sqrt{z_{\pmb \epsilon}(x_{m,k'})}\Big |^2\nonumber \\
& \hskip-0.08in \le  & \hskip-0.08in 4\sum_{m=1}^{2N-1}  |\pmb \epsilon(x_{m,k'})|\le 8N |{\pmb \epsilon}|,
\end{eqnarray*}
where  the second inequality  holds by  \eqref{phaselessminimization.1}, and the third estimate follows from the triangle inequality
\begin{equation} \label{halfinequality} |\sqrt{x^2+y}-|x| |\le \sqrt{|y|}\end{equation}
for all  $x\ge 0$ and $y\ge -x^2$.
Therefore there exist $1\le m_1, \ldots, m_N\le 2N-1$ and $\delta_{k'}\in \{-1, 1\}$ such that
\begin{equation*}
 \sum_{l=1}^{N} \Big(\sum_{k=k'L-N+1}^{k'L} \big(c_{\pmb \epsilon, k'}^0 (k)-\delta_{k'} c(k)\big) \phi(x_{m_l,k'}-k)\Big)^2
 \le 8N |{\pmb \epsilon}|.
\end{equation*}
This proves  \eqref{realother.thm.pf.eq3}.
\end{proof}

To prove Theorem \ref{realother.thm}, we next verify that for any $k'\in \Z$,
the  vector ${\bf c}_{\pmb \epsilon, k'}^{1/2}$
in the second step of the MEPS algorithm 
 is, up to a sign, not far away from ${\bf c}$ on $[k'L+1-N, k'L+(L-1)/2]$.

\begin{proposition}\label{c1/2.pr} Let ${\bf c}, {\pmb \epsilon}$ be as in Theorem \ref{realother.thm}, and
let vectors ${\bf c}_{\pmb \epsilon, k'}^{1/2}, k'\in \Z$, be as in  \eqref{c1/2.def}.
Then for any $k'\in\Z$,
there exists $\delta_{k'}\in \{-1, 1\}$ such that
\begin{equation}\label{c.half.error}
|c^{1/2}_{\pmb \epsilon,k'}(k)-\delta_{k'}c(k)|\le   \|(\Phi_N)^{-1}\|  (C_{f, \phi})^{k-k'L+N-1}\sqrt{8N |\pmb\epsilon|} 
\end{equation}
for all $k'L+1-N\le k\le  k'L+(L-1)/2$.
\end{proposition}

To prove Proposition \ref{c1/2.pr}, we need a technical lemma.

\begin{lemma}\label{lemma.D1} Let $h_2$ and $h_2^\ast$ be as in \eqref{gGamma2.def} and \eqref{gGamma2star.def}.
Then 
 \begin{equation}\label{nomiator.bound.eq1}
 |h_2({\bf e}_1, {\bf e}_2)| \le  \frac{\|{\bf e}_1\| \|{\bf e}_2\|}
{\min_{1\le n\le N} |\phi(\gamma_n)|}
\end{equation}
and
 \begin{equation}\label{nomiator.bound.eq2}
 |h_2^\ast ({\bf e}_1, {\bf e}_2)| \le  \frac{\|{\bf e}_1\| \|{\bf e}_2\|}
{\min_{1\le n\le N} |\phi(\gamma_n^*+N-1)|} 
\end{equation}
for all ${\bf e}_1, {\bf e}_2\in\R^N$.
\end{lemma}


\begin{proof} The upper bound estimate  \eqref{nomiator.bound.eq1} 
holds, since
\begin{eqnarray*}
|h_2({\bf e}_1,{\bf e}_2)|
\hskip-0.05in&\hskip-0.05in=&\hskip-0.1in \Big|\sum_{n=1}^{N}\Big(\frac{e_2(n)}{\phi(\gamma_n)}- \alpha\phi(\gamma_n)\Big)e_1(n)\Big|\\
\hskip-0.05in&\hskip-0.05in\le &\hskip-0.1in
\Big(\sum_{n=1}^{N}\Big|\frac{e_2(n)}{\phi(\gamma_n)}- \alpha \phi(\gamma_n)\Big|^2\Big)^{1/2}
\|{\bf e}_1\|\nonumber\\
\hskip-0.05in&\hskip-0.05in= &\hskip-0.1in
\Big(\sum_{n=1}^{N}\Big|\frac{e_2(n)}{\phi(\gamma_n)}\Big|^2 - \alpha^2 \sum_{n=1}^N|\phi(\gamma_n)|^2\Big)^{1/2}
\|{\bf e}_1\|\nonumber\\
\hskip-0.05in&\hskip-0.05in\le &\hskip-0.1in
\Big(\sum_{n=1}^{N}\Big|\frac{e_2(n)}{\phi(\gamma_n)}\Big|^2\Big)^{1/2}
\|{\bf e}_1\|,
\end{eqnarray*}
where $\alpha=\frac{\sum_{n=1}^{N}e_2(n)}{\sum_{n=1}^{N}|\phi(\gamma_n)|^2}$.

Applying similar argument, we can prove  \eqref{nomiator.bound.eq2}.
\end{proof}

Now we return to the proof of Proposition \ref{c1/2.pr}.

\begin{proof}[Proof of Proposition \ref{c1/2.pr}] Take $k'\in \Z$, and let
 ${\bf c}_{\pmb \epsilon, k', l}^{1/2},  0\le l\le (L-1)/2$, be as in \eqref{c.half.leastsquare}--\eqref{newhalf.eq3}.
 Observe that
  $$ c^{1/2}_{\pmb \epsilon,k'}(k)=c^{1/2}_{\pmb \epsilon,k', l}(k)$$ for all $k\in
  [k'L+1-N, k'L+(L-1)/2]$ and $l\ge \min(k-k'L+N-1, (L-1)/2)$.
Then it suffices to find $\delta_{k'}\in \{-1, 1\}$ such that
 \begin{eqnarray}\label{c.half.newerror}
 \hskip-0.28in & \hskip-0.08in  & \hskip-0.08in  \sum_{k=k'L+l+1-N}^{k'L+l}|c^{1/2}_{\pmb \epsilon,k', l}(k)-\delta_{k'}c(k)|^2\nonumber\\
\hskip-0.18in  &  \hskip-0.08in  &  \hskip0.08in \qquad \le   8N\|(\Phi_N)^{-1}\|^2  (C_{f, \phi})^{2l} |\pmb\epsilon|  \
\end{eqnarray}
 for all $0\le l\le (L-1)/2$.

 We establish the above conclusion \eqref{c.half.newerror} by induction.
 The conclusion \eqref{c.half.newerror} for $l=0$  follows from \eqref{realother.thm.pf.eq3}
 in  Proposition \ref{c0.pr}.
Inductively we assume that
\begin{eqnarray}\label{chalf.error.inductive}
 \hskip-0.28in & \hskip-0.08in  & \hskip-0.08in  \sum_{k=k'L+l_0+1-N}^{k'L+l_0}|c^{1/2}_{\pmb\epsilon,k', l_0}(k)-\delta_{k'}c(k)|^2\nonumber\\
\hskip-0.18in  &  \hskip-0.08in  &  \hskip0.08in \qquad
\le   8N\|(\Phi_N)^{-1}\|^2  (C_{f, \phi})^{2l_0} |\pmb\epsilon|
\end{eqnarray}
for some $0\le l_0\le (L-1)/2$.
Set
${\bf e}_{\pmb \epsilon, 1}= (\alpha_{\pmb\epsilon}(1), \ldots, \alpha_{\pmb\epsilon}(N))$ and
${\bf e}_{0, 1}= (\alpha(1), \ldots, \alpha(N))$,
where
$$\alpha_{\pmb\epsilon}(n)=\sum_{l'=0}^{N-2}c^{1/2}_{\pmb\epsilon,k', l_0}(k'L+l_0-l')\phi(\gamma_n+l'+1)$$
and
$$\alpha(n)=\sum_{l'=0}^{N-2}c(k'L+l_0-l')\phi(\gamma_n+l'+1), 1\le n\le N.$$
Now we divide  into two cases to prove  \eqref{c.half.newerror} for $l=l_0+1\le (L-1)/2$.

 {\sl Case 1}: $\sum_{k=k'L+l_0+2-N}^{k'L+l_0}|c(k)|^2=0$.

   Set
  $$\alpha_{\pmb \epsilon}(0)=\frac{\sum_{n=1}^N \phi(\gamma_n) \alpha_{\pmb \epsilon} (n)}
{\sum_{n=1}^N |\phi(\gamma_n)|^2}\ \ {\rm and} \ \  \alpha(0)=\frac{\sum_{n=1}^N \phi(\gamma_n) \alpha(n)}
{\sum_{n=1}^N |\phi(\gamma_n)|^2}.$$
Therefore  for the function $h_1$ in \eqref{gGamma.def}, we have
\begin{eqnarray} \label{denomiator.bound.pf.eq1}
\label{c1/2.lem.pf.eq1}
\hskip-0.13in  h_1({\bf e}_{\pmb\epsilon, 1})
 &\hskip-0.1in =& \hskip-0.08in
\sum_{n=1}^N (\alpha_{\pmb \epsilon}(n)- \phi(\gamma_n) \alpha_{\pmb \epsilon}(0))^{2}\nonumber
\\
\hskip-0.02in  & \hskip-0.08in \le  & \hskip-0.08in
\sum_{n=1}^N |\alpha_{\pmb \epsilon}(n)|^2\le \| \Phi\|^2
 \sum_{k=k'L+l_0+2-N}^{k'L+l_0}|c^{1/2}_{\pmb\epsilon,k', l_0}(k)|^2
 \nonumber \\
& \hskip-0.08in \le  & \hskip-0.08in
 8N\| \Phi\|^2\|(\Phi_N)^{-1}\|^2  (C_{f, \phi})^{2l_0} |\pmb\epsilon|\nonumber\\
& \hskip-0.08in \le  & \hskip-0.08in
\frac{\|\Phi\|^2 S_f}{2^4N^2 (C_{f, \phi})^{4N+L-5-2l_0}}\le M_0,
 \end{eqnarray}
where the third inequality follows from the inductive hypothesis \eqref{chalf.error.inductive} and the last two estimates hold by
\eqref{M0.def} and \eqref{realother.thm.eq2}.
Hence
\begin{equation}\label{c1case1.def}
c^{1/2}_{\pmb\epsilon,k', l_0+1}(l_0^*+1)=\frac{\sum_{n=1}^{N} |\phi(\gamma_n)| \sqrt{z_{\pmb\epsilon}(\gamma_n+l_0^*+1)} }{\sum_{n=1}^{N}|\phi(\gamma_n)|^2}
\end{equation}
by \eqref{c.half.leastsquare}, where $\ell_0^*=k'L+l_0$.

 {\sl Case 1a}:  $c(k'L+l_0+1)=0$.

 In this subcase,
$$z_{\pmb\epsilon}(\gamma_n+k'L+l_0+1)= \pmb \epsilon (\gamma_n+k'L+l_0+1), 1\le n\le N,$$
and
\begin{eqnarray} \label{case1.a}
 \hskip-0.18in & \hskip-0.08in &  \hskip-0.08in |c^{1/2}_{\pmb\epsilon,k', l_0+1}(k'L+l_0+1)-\delta_{k'} c(k'L+l_0+1)|\nonumber\\
\hskip-0.18in& \hskip-0.08in = &  \hskip-0.08in\frac{\sum_{n=1}^{N} |\phi(\gamma_n)| \sqrt{\pmb\epsilon (\gamma_n+k'L+l_0+1)} }{\sum_{n=1}^{N}|\phi(\gamma_n)|^2}\nonumber\\
\hskip-0.18in &\hskip-0.08in \le & \hskip-0.08in   \Big(\sum_{n=1}^{N}|\phi(\gamma_n)|^2\Big)^{-1/2}\hskip-0.08in \sqrt{N|{\pmb \epsilon}|}\nonumber\\
\hskip-0.18in &\hskip-0.08in \le & \hskip-0.08in 
\|(\Phi_N)^{-1}\| \sqrt{N|{\pmb \epsilon}|} .
\end{eqnarray}

 {\sl Case 1b}:   $c(k'L+l_0+1)\ne 0$.

 In this subcase,  it follows from 
 Theorem \ref{realX.thm}
that
$$c(k)=0 \  {\rm  for\ all} \ k\le k'L+l_0.$$
Therefore
the inductive hypothesis \eqref{c.half.newerror} holds for all $0\le l\le l_0$  with arbitrary $\delta_{k'}\in \{-1, 1\}$. So we may select
$$\delta_{k'}=\frac{c(k'L+l_0+1)}{|c(k'L+l_0+1)|}$$
in this subcase.
Hence
 \begin{eqnarray} \label{case1.b}
\hskip-0.18in & \hskip-0.08in & \hskip-0.08in |c^{1/2}_{\pmb\epsilon,k', l_0+1}(k'L+l_0+1)- \delta_{k'} c(k'L+l_0+1)|\nonumber\\
\hskip-0.18in &\hskip-0.08in \le & \hskip-0.08in\frac{\sum_{n=1}^{N}|\phi(\gamma_n)| \sqrt{|\pmb\epsilon (\gamma_n+k'L+l_0+1)| }}{\sum_{n=1}^{N}|\phi(\gamma_n)|^2}\nonumber\\
\hskip-0.18in &\hskip-0.08in \le & \hskip-0.08in \Big(\sum_{n=1}^{N}|\phi(\gamma_n)|^2\Big)^{-1/2} \hskip-0.08in \sqrt{N|{\pmb \epsilon}|}
\le \|(\Phi_N)^{-1}\| \sqrt{N|{\pmb \epsilon}|} ,
\end{eqnarray}
where the first equality follows from   \eqref{halfinequality}, \eqref{c1case1.def} and
$$f(\gamma_n+k'L+l_0+1)= c(k'L+l_0+1)\phi(\gamma_n), 1\le n\le N.$$
Thus for the Case 1, the estimates in \eqref{case1.a} and \eqref{case1.b}, together with the inductive hypothesis \eqref{chalf.error.inductive},
 imply 
%
\eqref{c.half.newerror} for $l=l_0+1$.

\medskip
 {\sl Case 2}: $\sum_{k=k'L+l_0+2-N}^{k'L+l_0}|c(k)|^2\ne 0$.

In this case,
\begin{equation}\label{eepsilon0.eq}
\|{\bf e}_{\pmb \epsilon, 1}-\delta_{k'} {\bf e}_{\pmb 0, 1}\|\le \|\Phi\|  \|(\Phi_N)^{-1}\|
  (C_{f, \phi})^{l_0}\sqrt{8N  |{\pmb \epsilon}|}
\end{equation}
by inductive hypothesis  \eqref{chalf.error.inductive}, and
\begin{eqnarray}\label{eepsilon0.eq2}
\|{\bf e}_{\pmb \epsilon, 1}\| & \hskip-0.08in \le & \hskip-0.08in  \|{\bf e}_{\pmb \epsilon, 1}-\delta_{k'} {\bf e}_{\pmb 0, 1}\|+\| {\bf e}_{\pmb 0, 1}\|\nonumber\\
& \hskip-0.08in \le & \hskip-0.08in  2 \|\Phi\| \Big(\sum_{k=k'L+l_0+2-N}^{k'L+l_0}|c(k)|^2\Big)^{1/2}
\end{eqnarray}
by \eqref{realother.thm.eq2} and the property that
\begin{equation}\|\Phi\| \|(\Phi_N)^{-1}\|\ge 1.\end{equation}
Therefore
\begin{eqnarray} \label{c1/2.lem.pf.eq1}
\hskip-0.33in  h_1({\bf e}_{\pmb\epsilon, 1})
 &\hskip-0.1in =& \hskip-0.08in
\sum_{n=1}^N (\alpha_{\pmb \epsilon}(n)- \phi(\gamma_n) \alpha_{\pmb \epsilon}(0))^{2}\nonumber
\\
\hskip-0.33in&\hskip-0.10in \ge& \hskip-0.08in\|(\Phi_N)^{-1}\|^{-2}   
\hskip-0.08in  \sum_{k=k'L+l_0+2-N}^{k'L+l_0}|c^{1/2}_{\pmb\epsilon,k'}(k)|^2
\nonumber\\
\hskip-0.33in &\hskip-0.10in \ge& \hskip-0.08in \frac{\sum_{k=k'L+l_0+2-N}^{k'L+l_0}|c(k)|^2}{4\|(\Phi_N)^{-1}\|^{2} } 
 > M_0,
\end{eqnarray}
where the second inequality follows from \eqref{realother.thm.eq2} and the inductive hypothesis \eqref{chalf.error.inductive}, and
the last inequality holds by \eqref{sf.def} and \eqref{M0.def}.
Hence
\begin{equation} \label{c1/2.lem.pf.eq2}
d_{\pmb\epsilon,k'}(k'L+l_0+1)=\frac{ h_2({\bf e}_{\pmb \epsilon, 1}, {\bf e}_{\pmb \epsilon, 2})}{2 h_1({\bf e}_{\pmb\epsilon, 1})}
\end{equation}
by  \eqref{phaselessminimization.1n}, where
${\bf e}_{\pmb \epsilon, 2}= (\eta_{\pmb \epsilon}(1), \ldots, \eta_{\pmb \epsilon}(N))$
with
$$\eta_{\pmb \epsilon}(n)= z_{\pmb\epsilon}(\gamma_n+k'L+l_0+1)-|\alpha_{\pmb\epsilon}(n)|^2, \ 1\le n\le N.$$
Set
${\bf e}_{\pmb 0, 2}= (\eta(1), \ldots, \eta(N))$,
where
$$\eta(n)  =   |f(\gamma_n+k'L+l_0+1)|^2-|\alpha(n)|^2, 1\le n\le N.$$
Then it follows from \eqref{cp.def} that
\begin{equation} \label{c1/2.lem.pf.eq2+}
c(k'L+l_0+1)=\frac{ h_2({\bf e}_{0, 1}, {\bf e}_{0, 2})}{2 h_1({\bf e}_{0, 1})}.
\end{equation}

To estimate  $|d_{\pmb\epsilon,k'}(k'L+l_0+1)- \delta_{k'} c(k'L+l_0+1)|$, we set
$$\beta_{\pmb \epsilon}(n)= \alpha_{\pmb \epsilon}(n)-\phi(\gamma_n)\alpha_{\pmb \epsilon}(0) \ \ {\rm and}\
\  \beta_{0}(n)= \alpha(n)-\phi(\gamma_n)\alpha(0)$$
for $1\le n\le N$. Then
\begin{eqnarray}\label{beta.temp}
\hskip-0.1in  &\hskip-0.1in   & \hskip-0.1in  \sum_{n=1}^N |\beta_{\pmb \epsilon}(n)-\delta_{k'} \beta_{0}(n)|^2\nonumber\\
\hskip-0.1in  & \hskip-0.1in \le &  \hskip-0.1in \sum_{n=1}^N \Big|\sum_{l=0}^{N-2} \Big(c_{\pmb \epsilon, k', l_0}^{1/2}(k'L+l_0-l)\nonumber\\
\hskip-0.1in  & \hskip-0.1in  & -\delta_{k'} c(k'L+l_0-l)\Big) \phi(\gamma_n+l'+1)\Big|^2\nonumber\\
\hskip-0.1in &\hskip-0.1in  \le &  \hskip-0.1in \|\Phi\|^2\Big(
\sum_{k=k'L+l_0+2-N}^{k'L+l_0}|c^{1/2}_{\pmb\epsilon,k', l_0}(k)-\delta_{k'}c(k)|^2 
\Big).
\end{eqnarray}
Hence
\begin{eqnarray} \label{c1/2.lem.pf.eq1+}
\hskip-0.08in & \hskip-0.08in& \hskip-0.08in| h_1({\bf e}_{\pmb\epsilon, 1})- h_1({\bf e}_{0, 1})|\nonumber\\
\hskip-0.08in & \hskip-0.08in=& \hskip-0.08in
\Big|\sum_{n=1}^N (\beta_{\pmb \epsilon}(n)-  \delta_{k'} \beta_0(n)) (\beta_{\pmb \epsilon}(n)+ \delta_{k'} \beta_0(n))\Big|\nonumber\\
\hskip-0.08in & \hskip-0.08in\le& \hskip-0.08in
\|\Phi\|^2 \Big(\sum_{k=k'L+l_0+2-N}^{k'L+l_0}|c^{1/2}_{\pmb\epsilon,k', l_0}(k)-\delta_{k'}c(k)|^2\Big)^{1/2}\nonumber\\
\hskip-0.08in & \hskip-0.08in & \hskip-0.08in\times
\Big(\sum_{k=k'L+l_0+2-N}^{k'L+l_0}|c^{1/2}_{\pmb\epsilon,k', l_0}(k)+\delta_{k'}c(k)|^2\Big)^{1/2}\nonumber\\
 \hskip-0.08in & \hskip-0.08in\le& \hskip-0.08in  3 \sqrt{8N |{\pmb \epsilon}|} \|\Phi\|^2
\|(\Phi_{N})^{-1}\| (C_{f, \phi})^{l_0/2}\nonumber\\
& & \times \Big(\sum_{k=k'L+l_0+2-N}^{k'L+l_0}|c(k)|^2\Big)^{1/2},
\end{eqnarray}
 where the equality is true by  the equality in  \eqref{denomiator.bound.pf.eq1},
 the  first inequality  holds by \eqref{beta.temp},  and
the last inequality  follows from \eqref{realother.thm.eq2} and  the inductive hypothesis \eqref{chalf.error.inductive}. 

Observe that
\begin{equation}\label{e02.estimate}
\|{\bf e}_{\pmb 0, 2}\| \le 
\sum_{n=1}^N |\eta(n)|
\le 2\|\Phi\|^2\Big(\sum_{k=k'L+l_0+2-N}^{k'L+l_0+1}|c(k)|^2
\Big)
\end{equation}
and
\begin{eqnarray}\label{eepsilone0.estimate}
& \hskip-0.08in  &  \hskip-0.08in  \|{\bf e}_{\pmb \epsilon, 2}- {\bf e}_{\pmb 0, 2}\|\le   \sum_{n=1}^N |\eta_{\pmb \epsilon}(n)- \eta(n)|\nonumber\\
 & \hskip-0.08in \le   & \hskip-0.08in  N |{\pmb \epsilon}|+ \|\Phi\|^2
\Big(\sum_{k=k'L+l_0+2-N}^{k'L+l_0}|c^{1/2}_{\pmb\epsilon,k', l_0}(k)-\delta_{k'}c(k)|^2\Big)^{1/2}\nonumber\\
& \hskip-0.08in   &  \hskip-0.08in \times
\Big(\sum_{k=k'L+l_0+2-N}^{k'L+l_0}|c^{1/2}_{\pmb\epsilon,k', l_0}(k)+\delta_{k'}c(k)|^2\Big)^{1/2}
\nonumber\\
& \hskip-0.08in  \le  &  \hskip-0.08in N |{\pmb \epsilon}|+ 3\|\Phi\|^2 \|(\Phi_N)^{-1}\|
\sqrt{8N  |{\pmb\epsilon}}| (C_{f, \phi})^{l_0}\nonumber\\
& \hskip-0.08in   &  \hskip-0.08in \times
\Big(\sum_{k=k'L+l_0+2-N}^{k'L+l_0}|c(k)|^2\Big)^{1/2}\nonumber\\
& \hskip-0.08in  \le  &  \hskip-0.08in 4 \|\Phi\|^2 \|(\Phi_N)^{-1}\|
 \sqrt{8N  |{\pmb \epsilon}}| (C_{f, \phi})^{l_0}
\nonumber\\
& \hskip-0.08in   &  \hskip-0.08in \times
 \Big(\sum_{k=k'L+l_0+2-N}^{k'L+l_0}|c(k)|^2\Big)^{1/2},
\end{eqnarray}
where the third inequality follows from the inductive hypothesis \eqref{chalf.error.inductive} and the last one holds by \eqref{realother.thm.eq2}.
Therefore we get from  \eqref{eepsilon0.eq}, \eqref{eepsilon0.eq2},  \eqref{e02.estimate}, \eqref{eepsilone0.estimate} and Lemma \ref{lemma.D1}
that
\begin{eqnarray}\label{h2.difference}
&\hskip-0.08in &  \hskip-0.08in|h_2({\bf e}_{\pmb\epsilon,1}, {\bf e}_{\pmb\epsilon,2})-\delta_{k'} h_2({\bf  e}_{0,1}, {\bf e}_{0,2})|\nonumber\\
& \hskip-0.08in \le & \hskip-0.08in |h_2({\bf e}_{\pmb\epsilon,1}, {\bf e}_{\pmb\epsilon,2}-{\bf e}_{0,2})|+|h_2({\bf  e}_{\pmb\epsilon,1}- \delta_{k'}{\bf e}_{0,1}, {\bf e}_{0,2})|\nonumber\\
&\hskip-0.08in \le&  \hskip-0.08in
 \frac{\|{\bf e}_{\pmb\epsilon,1}\|\|{\bf e}_{\pmb\epsilon,2}-{\bf e}_{0,2}\| +\|{\bf  e}_{\pmb\epsilon,1}-\delta_{k'}{\bf e}_{0,1}\|\|{\bf e}_{0,2}\|}
{\min_{1\le n\le N} |\phi(\gamma_n)|}
\nonumber\\
&  \hskip-0.08in \le& \hskip-0.08in
\frac{  10 \|\Phi\|^3 \|(\Phi_N)^{-1}\|}
{\min_{1\le n\le 2N-1} |\phi(\gamma_n)|}\sqrt{8N  |{\pmb \epsilon}}| (C_{f, \phi})^{l_0}
\nonumber\\
& \hskip-0.08in   &  \hskip-0.08in \times
 \Big(\sum_{k=k'L+l_0+2-N}^{k'L+l_0+1}|c(k)|^2\Big).
\end{eqnarray}
Hence
\begin{eqnarray}\label{dk'.estimate}
\hskip-0.08in &\hskip-0.08in &  \hskip-0.08in |d_{\pmb\epsilon,k'}(k'L+l_0+1)- \delta_{k'} c(k'L+l_0+1)|\nonumber\\
\hskip-0.08in & \hskip-0.08in \le &  \hskip-0.08in \frac{|h_2({\bf  e}_{\pmb\epsilon,1}, {\bf e}_{\pmb\epsilon,2})-\delta_{k'} h_2({\bf e}_{0,1}, {\bf e}_{0,2})|}{2 h_1({\bf e}_{\pmb\epsilon, 1})}\nonumber\\
\hskip-0.08in& \hskip-0.08in & \hskip-0.08in + \frac{|h_1({\bf e}_{\pmb\epsilon, 1})-h_1({\bf e}_{0, 1})|}{ h_1({\bf e}_{\pmb\epsilon, 1})}|c(k'L+l_0+1)|\nonumber\\
\hskip-0.08in &  \hskip-0.08in \le&  \hskip-0.08in 
\Big(\frac{  20 \|\Phi\|}
{\min_{1\le n\le N} |\phi(\gamma_n)|}+12\Big)  \|\Phi\|^2 \|(\Phi_N)^{-1}\|^3
\nonumber \\
\hskip-0.08in &   \hskip-0.08in&  \hskip-0.08in\times  \sqrt{8N  |{\pmb \epsilon}|} (C_{f, \phi})^{l_0}
\frac{\sum_{k=k'L+l_0+2-N}^{k'L+l_0+1}|c(k)|^2}{\sum_{k=k'L+l_0+2-N}^{k'L+l_0}|c(k)|^2} \nonumber\\
\hskip-0.08in &  \hskip-0.08in \le &
\hskip-0.08in\frac{  2^5 \|\Phi\|^3\|(\Phi_N)^{-1}\|^3 M_f }
{\min_{1\le n\le N} |\phi(\gamma_n)|}
 \sqrt{8N  |{\pmb \epsilon}|} (C_{f, \phi})^{l_0},\quad
\end{eqnarray}
where the first inequality follows from  \eqref{c1/2.lem.pf.eq2} and \eqref{c1/2.lem.pf.eq2+},
and the second estimate holds by  \eqref{c1/2.lem.pf.eq1}, \eqref{c1/2.lem.pf.eq1+} and \eqref{h2.difference}.
By \eqref{dk'.estimate} and the inductive hypothesis  \eqref{chalf.error.inductive}, we obtain
\begin{eqnarray}
 & & \Big| \sum_{m=1}^{N-1} c^{1/2}_{\pmb\epsilon, k', l_0}(k'L+l_0-m)\phi(\gamma_n+m)\nonumber\\
& &  + d_{\pmb\epsilon,k'}(k'L+l) \phi(\gamma_n)- \delta_{k'}
f(\gamma_n+k'L+l_0)\Big|^2\nonumber\\
& \le & \Big(\sum_{m=0}^{N-1} |\phi(\gamma_n+m)|^2\Big)
\nonumber\\
 & & \times\Big(
\sum_{m=1}^{N-1} |c^{1/2}_{\pmb\epsilon, k', l_0}(k'L+l_0-m)- \delta_{k'} c(k'L+l_0-m)|^2\nonumber\\
& & +
|d_{\pmb\epsilon,k'}(k'L+l_0+1)- \delta_{k'} c(k'L+l_0+1)|^2\Big)\nonumber\\
& \le  &
 \frac{  2^{14}  N \|\Phi\|^8\|(\Phi_N)^{-1}\|^6 (M_f)^2 }
{\min_{1\le n\le N} |\phi(\gamma_n)|^2} (C_{f, \phi})^{2l_0}  |{\pmb \epsilon}|
  \end{eqnarray}
  for all $1\le n\le N$.
  This implies that,
for those $1\le n\le N$ satisfying
 \begin{equation}\label{condition1}
 |f(\gamma_n+k'L+l_0)|>  \frac{  2^{7}   \|\Phi\|^4\|(\Phi_N)^{-1}\|^3 M_f }
{\min_{1\le n\le N} |\phi(\gamma_n)|} (C_{f, \phi})^{l_0}  \sqrt{N |{\pmb \epsilon}|}, \end{equation}
the sign  $\delta_l(n)$ in \eqref{mewhalf.eq2}
  are the same as the one of $\delta_{k'}
f(\gamma_n+k'L+l_0)$, and hence
 $\tilde {z}_{\pmb \epsilon}(k'L+l+\gamma_n)$ in  \eqref{newhalf.eq3}
 satisfies
 \begin{equation}\label{lasttildez.eq1}
 |\tilde {z}_{\pmb \epsilon}(k'L+l+\gamma_n)-\delta_{k'}
f(\gamma_n+k'L+l_0)|\le \sqrt{|\pmb \epsilon|}.
 \end{equation}
For those $1\le n\le N$ such that \eqref{condition1} fails, 
\begin{eqnarray}\label{lasttildez.eq2}
 & & |\tilde {z}_{\pmb \epsilon}(k'L+l+\gamma_n)-\delta_{k'}
f(\gamma_n+k'L+l_0)| \nonumber\\
&  \le &
|\tilde {z}_{\pmb \epsilon}(k'L+l+\gamma_n)|+|f(\gamma_n+k'L+l_0)|\nonumber \\
& \le &
\frac{ 2^{9}   \|\Phi\|^4\|(\Phi_N)^{-1}\|^3 M_f }
{\min_{1\le n\le N} |\phi(\gamma_n)|} (C_{f, \phi})^{l_0}  \sqrt{N |{\pmb \epsilon}|}. 
\end{eqnarray}
Combining \eqref{newhalf.eq3}, \eqref{lasttildez.eq1} and \eqref{lasttildez.eq2}, we get
\begin{eqnarray}\label{chalf.error.inductivenew}
 \hskip-0.28in & \hskip-0.08in  & \hskip-0.08in  \sum_{k=k'L+l_0+2-N}^{k'L+l_0+1}|c^{1/2}_{\pmb\epsilon,k', l_0+1}(k)-\delta_{k'}c(k)|^2\nonumber\\
\hskip-0.18in  &  \hskip-0.08in\le   &  \hskip-0.08in
  \frac{ 2^{18}  N \|\Phi\|^8\|(\Phi_N)^{-1}\|^8 (M_f)^2 }
{\min_{1\le n\le N} |\phi(\gamma_n)|^2} (C_{f, \phi})^{2l_0}  N |{\pmb \epsilon}|\nonumber\\
\hskip-0.18in  &  \hskip-0.08in\le   &  \hskip-0.08in 8N\|(\Phi_N)^{-1}\|^2  (C_{f, \phi})^{2l_0+2} |\pmb\epsilon|.
\end{eqnarray}
Thus the inductive proof can proceed for the Case 2. This  completes the  proof.
\end{proof}

To prove Theorem \ref{realother.thm}, we then justify that for any $k'\in \Z$,
 the vector ${\bf c}_{\pmb \epsilon, k'}^{1}$
in the  third step of the MEPS algorithm 
 is, up to a sign, not far away from ${\bf c}$ on $[k'L+1-N-(L-1)/2, k'L+(L-1)/2]$.

\begin{proposition}\label{c1.pr} Let ${\bf c}, {\pmb \epsilon}$ be as in Theorem \ref{realother.thm},
and let
vectors ${\bf c}_{\pmb \epsilon, k'}^{1}, k'\in \Z$, be as in  \eqref{c1.def}.
Then for any $k'\in\Z$
there exists $\delta_{k'}\in \{-1, 1\}$ such that
\begin{equation}\label{cone.error}
|c^{1}_{\pmb \epsilon,k'}(k)-\delta_{k'}c(k)|\le  N \|(\Phi_N)^{-1}\|  (C_{f, \phi})^{|k-k'L|+N-1}\sqrt{8|\pmb\epsilon|} 
\end{equation}
for all $k'L+1-N-(L-1)/2\le k\le  k'L+(L-1)/2$. 
\end{proposition}

\begin{proof}
 Let
 ${\bf c}_{\pmb \epsilon, k', l'}^{1/2}, k'\in \Z, 0\le l'\le (L-1)/2$, be as in \eqref{III.11}--\eqref{newone.eq3}.
 Observe that
  $$c^{1}_{\pmb \epsilon,k'}(k)= c^{1}_{\pmb \epsilon,k', l'}(k)$$ for all $k\in
  [k'L+1-N-l', k'L+(L-1)/2]$ and $l'\ge \min(k'L-k, (L-1)/2)$.
 Then it suffices to prove
 \begin{eqnarray}\label{c.one.newerror}
 \hskip-0.28in & \hskip-0.08in  & \hskip-0.08in  \sum_{k=k'L-l'+1-N}^{k'L-l'}|c^{1}_{\pmb \epsilon,k', l'}(k)-\delta_{k'}c(k)|^2\nonumber\\
\hskip-0.18in  &  \hskip-0.08in  &  \hskip0.08in \qquad \le   8N^2\|(\Phi_N)^{-1}\|^2  (C_{f, \phi})^{2(l'+N-1)} |\pmb\epsilon|  \
\end{eqnarray}
 by induction on $0\le l'\le (L-1)/2$.
 The conclusion \eqref{c.one.newerror} for $l'=0$ holds by Proposition  \ref{c1/2.pr}.
Similar to argument used to prove \eqref{c.half.newerror}, we can show that
\eqref{c.one.newerror} holds for any $0<l'\le (L-1)/2$ by induction.
\end{proof}

To prove Theorem \ref{realother.thm}, we finally  adjust phases of
 ${\bf c}_{\pmb \epsilon, k'}^{1}, k'\in \Z$,
in the fourth  step of the MEPS algorithm. 

\begin{proposition}\label{c2.pr}
Let  ${\bf c}_{\pmb \epsilon, k'}^{1}$ and $\delta_{k'}\in \{-1, 1\}, k'\in \Z$, be as in \eqref{c1/2.def} and  Proposition \ref{c1.pr} respectively.
If \eqref{realother.thm.eq2} holds,  there exists $\delta\in \{-1, -1\}$ such that
\begin{equation}\label{c2.lem.eq222}
{\bf c}_{\pmb \epsilon, k'}^2= \delta \delta_{k'} {\bf c}_{\pmb \epsilon, k'}^1
\end{equation}
for all $k'\in \Z$ with $\sum_{k=2-N}^{0} |c(k-k'L-(L-1)/2)|^2\ne 0$.
\end{proposition}

\begin{proof}
By \eqref{c.half.leastsquare}, \eqref{phaselessminimization.1n},
\eqref{III.11} and \eqref{phaselessminimization.1n++},
we have
$$\langle c_{\pmb \epsilon, k'}^1,  c_{\pmb \epsilon, k'+1}^1\rangle=
\sum_{k=k'L+(L-1)/2-N+2}^{k'L+(L-1)/2} c_{\pmb \epsilon, k'}^1(k) c_{\pmb \epsilon, k'+1}^1(k).$$
Therefore   for any  $k'\in \Z$,
\begin{eqnarray*} \label{realother.thm.pf.eq4}
&\hskip-0.08in  &  \hskip-0.08in \Big|\langle  \delta_{k'} c_{\pmb \epsilon, k'}^1, \delta_{k'+1}  c_{\pmb \epsilon, k'+1}^1\rangle - \sum_{k=k'L+(L-1)/2-N+2}^{k'L+(L-1)/2} |c(k)|^2\Big|
\nonumber\\
& \hskip-0.08in \le &  \hskip-0.08in \sum_{k=k'L+(L-1)/2-N+2}^{k'L+(L-1)/2}
|\delta_{k'}  c_{\pmb \epsilon, k^\prime}^{1}(k)-c(k)| |c(k)| \nonumber\\
& \hskip-0.08in & \hskip-0.08in
+  \sum_{k=k'L+(L-1)/2-N+2}^{k'L+(L-1)/2} 
|\delta_{k'+1}c^1_{\pmb \epsilon, k^\prime+1}(k)-c(k)| |c_{\pmb \epsilon, k'}^1(k)|\nonumber\\
& \hskip-0.08in \le &  \hskip-0.08in  3 \sqrt{ 8N^3 |\pmb\epsilon |} \|(\Phi_N)^{-1}\|  (C_{f, \phi})^{2(N-1)+(L-1)/2} \nonumber\\
& \hskip-0.08in & \hskip-0.08in \times \Big(\sum_{k=k'L+(L-1)/2-N+2}^{k'L+(L-1)/2} |c(k)|^2\Big)^{1/2}\nonumber\\
&  \hskip-0.08in \le  &  \hskip-0.08in  \frac{3}{4} \Big(\sum_{k=k'L+(L-1)/2-N+2}^{k'L+(L-1)/2} |c(k)|^2\Big),
\end{eqnarray*}
where the second estimate follows from Proposition \ref{c1.pr}, and the last inequality holds by the assumption \eqref{realother.thm.eq2}
on the noise level $|{\pmb \epsilon}|$.
Therefore the vectors $\delta_{k'} c_{\pmb \epsilon, k'}^1$ and $\delta_{k'}c_{\pmb \epsilon, k'+1}^1$  have positive inner product.
This together Theorem \ref{realX.thm} proves \eqref{c2.lem.eq222}.
 \end{proof}

We finish this section with the proof of Theorem \ref{realother.thm}.

\begin{proof}[Proof of Theorem \ref{realother.thm}] Let
$$k'_+= \lfloor (K_+(f)+(L-1)/2)/L\rfloor$$
and
$$k'_-= \lfloor (K_-(f)-(L-1)/2)/L\rfloor,
$$
 and   set $ k'= \lfloor(2k+L-1)/(2L)\rfloor, k\in \Z$.
  Then for $k\in [k'_- L-(L-1)/2, k'_+L+(L-1)/2]$,
  we obtain from  \eqref{phaselessminimization.2}, \eqref{phaselessminimization.3}, \eqref{phaselessminimization.4}, and  Propositions \ref{c1.pr} and \ref{c2.pr} that
\begin{eqnarray}\label{realother.thm.pf.eq1}
& \hskip-0.08in & \hskip-0.08in |c_{\pmb \epsilon}(k)-\delta c(k)|  =   |\delta \delta_{k'} c^1_{\pmb \epsilon, k'}(k)-\delta c(k)|
\nonumber\\
& \hskip-0.08in \le & \hskip-0.08in  N\|(\Phi_N)^{-1}\| (C_{f, \phi})^{N-1+(L-1)/2} \sqrt{8|\pmb\epsilon|},
\end{eqnarray}
 where $\delta\in \{-1, 1\}$ is given in Proposition \ref{c2.pr}.
Observe that  $c(k)=0$ for all $k\not\in [k'_- L-(L-1)/2, k'_+L+(L-1)/2]$.
Thus for $k\not\in [k'_- L-(L-1)/2, k'_+L+(L-1)/2]$,
\begin{eqnarray}\label{realother.thm.pf.eq2}
& \hskip-0.08in & \hskip-0.08in |c_{\pmb \epsilon}(k)-\delta c(k)| =|c_{\pmb \epsilon}(k)| =   |c^1_{\pmb \epsilon, k'}(k)-\delta_{k'} c(k)| \nonumber\\
& \hskip-0.08in \le & \hskip-0.08in  N \|(\Phi_N)^{-1}\| (C_{f, \phi})^{N-1+(L-1)/2} \sqrt{8|\pmb\epsilon|}.
\end{eqnarray}
by Proposition  \ref{c1.pr}.
Combining \eqref{realother.thm.pf.eq1} and \eqref{realother.thm.pf.eq2}
completes the proof.
\end{proof}

\end{appendices}

\ifCLASSOPTIONcompsoc
  \section*{Acknowledgments}
\else
  \section*{Acknowledgment}
\fi

The authors would like to thank Professor
Zhiqiang Xu for his comments and suggestions for the improvement of this manuscript.  

 The project is partially supported by National Science Foundation (DMS-1412413).

\end{document}